\documentclass[oneside,11pt,letterpaper]{amsart}

\usepackage[usenames,dvipsnames]{xcolor}
\usepackage[colorlinks=true,linkcolor=Blue,citecolor=NavyBlue]{hyperref}
\usepackage{amsmath,amssymb,epsf}
\usepackage{amsfonts,amsbsy,amscd,stmaryrd}
\usepackage{latexsym,amsopn}
\usepackage{amsthm,mathtools}
\usepackage{mathrsfs}
\usepackage{slashed}

\newcommand{\ep}{\epsilon}

\allowdisplaybreaks

\usepackage{graphicx}

\usepackage{color}
\definecolor{Red}{rgb}{1.,0.,0.}
\newcounter{smallarabics}
\newenvironment{arabicenumerate}
{\begin{list}{{\normalfont\textrm{(\arabic{smallarabics})}}}
  {\usecounter{smallarabics}\setlength{\itemindent}{0cm}
   \setlength{\leftmargin}{5ex}\setlength{\labelwidth}{4ex}
   \setlength{\topsep}{0.75\parsep}\setlength{\partopsep}{0ex}
   \setlength{\itemsep}{0ex}}}
{\end{list}}

\newcommand*{\wbar}[1]{\overbracket[0.7pt][-2pt]{#1\!}\,}

\newcommand{\ben}{\begin{arabicenumerate}}  
\newcommand{\een}{\end{arabicenumerate}}

\newtheorem{theorem}{Theorem}[section]

\newtheorem{proposition}[theorem]{Proposition}
\newtheorem{lemma}[theorem]{Lemma}
\newtheorem{corollary}[theorem]{Corollary}

\theoremstyle{definition}
\newtheorem{definition}[theorem]{Definition}
\newtheorem{remark}[theorem]{Remark}
\newtheorem{example}[theorem]{Example}
\newcommand{\beq}{\begin{equation}}
\newcommand{\eeq}{\end{equation}}
\newcommand{\bea}{\begin{aligned}}
\newcommand{\eea}{\end{aligned}}
\newcommand{\bex}{\begin{example}}
\newcommand{\eex}{\end{example}}
\def\bel{\begin{lemma}}
\def\eel{\end{lemma}}
\def\bet{\begin{theorem}}
\def\eet{\end{theorem}}
\def\bed{\begin{definition}}
\def\eed{\end{definition}}
\def\ber{\begin{remark}}
\def\eer{\end{remark}}
\newcommand{\qeds}{\qed\medskip}
\def\proof{
\noindent{\it Proof.}\ \ }

\def\rr{{\mathbb R}}
\def\zz{{\mathbb Z}}
\def\cc{{\mathbb C}}
\def\nn{{\mathbb N}}

\def\ss{{\mathbb S}}

\usepackage{calrsfs}
\DeclareMathAlphabet{\pazocal}{OMS}{zplm}{m}{n}

\def\cD{{\pazocal D}}
\def\cE{\pazocal{E}}

\def\cH{{\pazocal H}}

\def\cU{{\pazocal U}}

\def\cM{{\pazocal M}}

\def\cY{{\pazocal Y}}
\def\cZ{{\pazocal Z}}
\def\cV{{\pazocal V}}

\newcommand{\Yspace}{\pazocal{Y}}

\newcommand{\pa}{\partial}

\newcommand{\RR}{\mathbb{R}}

\newcommand{\abs}[1]{{\lvert{#1}\rvert}}
\newcommand{\ang}[1]{{\langle{#1}\rangle}}

\def\c{{\rm c}}

\def\sp{{\rm sp}}

\def\loc{{\rm loc}}

\let\Im\relax
\let\Re\relax

\DeclareMathOperator{\Im}{Im}
\DeclareMathOperator{\Re}{Re}
\DeclareMathOperator{\Dom}{Dom}
\DeclareMathOperator{\Tr}{Tr}

\DeclareMathOperator{\Ran}{Ran}

\DeclareMathOperator{\supp}{supp}

\def \p{ \partial}

\def\14{\frac{1}{4}}

\def\12{\frac{1}{2}}

\newcommand{\one}{\boldsymbol{1}}

\def\bep{\begin{proposition}}
\def\eep{\end{proposition}}
\def\b{{\rm b}}

\newcommand{\bra}{\langle} 
\newcommand{\ket}{\rangle}
\newcommand{\leftw}{{:}}
\newcommand{\rightw}{{:}}
\renewcommand{\geq}{\geqslant}
\renewcommand{\leq}{\leqslant}

\makeatletter
\newsavebox\myboxA
\newsavebox\myboxB
\newlength\mylenA

\newcommand*\xoverline[2][0.75]{%
    \sbox{\myboxA}{$\m@th#2$}%
    \setbox\myboxB\null
    \ht\myboxB=\ht\myboxA%
    \dp\myboxB=\dp\myboxA%
    \wd\myboxB=#1\wd\myboxA
    \sbox\myboxB{$\m@th\overline{\copy\myboxB}$}
    \setlength\mylenA{\the\wd\myboxA}
    \addtolength\mylenA{-\the\wd\myboxB}%
    \ifdim\wd\myboxB<\wd\myboxA%
       \rlap{\hskip 0.5\mylenA\usebox\myboxB}{\usebox\myboxA}%
    \else
        \hskip -0.5\mylenA\rlap{\usebox\myboxA}{\hskip 0.5\mylenA\usebox\myboxB}%
    \fi}
\makeatother

\newcommand*{\defeq}{:=}

\def\WF{{\rm WF}}

\def\cf{C^\infty}

\def\zero{{\rm\textit{o}}}
\def\be{{}^{\rm b}}

\def\inti{{\circ}}

\def\wf{{\rm WF}}

\newcommand{\orega}[1]{{\cE'(M)\to \cf(M)}}

\def\loc{{\rm loc}}

\DeclareMathOperator{\sgn}{sgn}

\newcommand{\open}[1]{\mathopen{}\mathclose{\left]#1 \right[}}
\newcommand{\clopen}[1]{\mathopen{}\mathclose{\left[#1 \right[}}
\newcommand{\opencl}[1]{\mathopen{}\mathclose{\left]#1 \right]}}
\newcommand{\opens}[1]{\mathopen{}\mathclose{]#1 [}}
\newcommand{\clopens}[1]{\mathopen{}\mathclose{[#1 [ }}
\newcommand{\opencls}[1]{\mathopen{}\mathclose{]#1 ]}}

\newcommand{\ccf}{C_{\rm c}^\infty}

\newcommand{\norm}[1]{\left\|{#1}\right\|}
\newcommand{\normm}[1]{\big\|{#1}\big\|}

\newcommand{\cosb}{\be S^*M}

\newcommand{\Mwithvalues}{M;\cc^4}

\newcommand{\Hlinf}{H_{\b,\loc}^{1,\infty}(\Mwithvalues)}
\newcommand{\Hcminfd}{H_{\b,\c}^{-1,-\infty}(\Mwithvalues)}

\newcommand{\dHcminf}{H_{\b,\c}^{1,-\infty}(\Mwithvalues)}
\newcommand{\dHlminf}{H_{\b,\loc}^{1,-\infty}(\Mwithvalues)}
\newcommand{\dHcinf}{H_{\b,\c}^{1,\infty}(\Mwithvalues)}
\newcommand{\dHlinf}{H_{\b,\loc}^{1,\infty}(\Mwithvalues)}

\newcommand{\Hl}[1]{H_{\b,\loc}^{1,#1-1}(\Mwithvalues)}
\newcommand{\Hcd}[1]{H_{\b,\c}^{-1,#1+1}(\Mwithvalues)}

\newcommand{\cotwice}{\cosb\times\cosb}

\def\st{{ \ |\  }}

\newcommand{\module}[1]{\left|#1\right|}

\newcommand{\Sig}{\Sigma}
\newcommand{\dSig}{\dot\Sigma}
\newcommand{\Sigin}{\dot\Sigma_{\rm in}}
\newcommand{\Sigout}{\dot\Sigma_{\rm out}}

\author{Dean Baskin}
\address{Department of Mathematics, Texas A\&M University, College Station, TX, USA}
\email{dbaskin@math.tamu.edu}

\author{Micha{\l} Wrochna} 
\address{Mathematical Institute, Universiteit Utrecht, Utrecht, The Netherlands \vspace{-0.3cm}} \address{Mathematics \& Data Science, Vrije Universiteit Brussel, Brussels, Belgium}
 \email{{m.wrochna@uu.nl}}

\author{Jared Wunsch}
\address{Department of Mathematics, Northwestern University, Evanston, IL, USA}
\email{jwunsch@math.northwestern.edu}

\begin{document}

\title{Singularities of Dirac--Coulomb propagators}

\begin{abstract}
In this paper we study singularities of propagators and $N$-point functions for Dirac fields in a Coulomb potential, possibly with a $t$-dependent smooth part for $|t|<T<\infty$. We show  that the \emph{in} and \emph{out} Dirac--Coulomb vacua are Hadamard states for $r\neq 0$. Furthermore, we prove that the relative charge density of any two Hadamard states is well-defined as a locally integrable function including near $r=0$. The results are based on a diffractive propagation of singularities theorem for the Dirac--Coulomb system previously obtained by the first and third authors, generalized here to the case of $t$-dependent potentials.
\end{abstract}

\maketitle

\section{Introduction and main result}

\subsection{Introduction}  Second quantization of Dirac fields in external potentials and the construction of perturbative series in interacting Quantum Field  Theory require a precise description of singularities of propagators and $N$-point functions. In  related models in relativistic Quantum Mechanics it is often possible to introduce  an ultra-violet cutoff that  improves the short-distance  behavior (see e.g.~\cite{Fournais2020,Sere2022} for the Dirac--Fock model, and \cite{Hainzl2005,Hainzl2007} and references therein for the Bogoliubov--Dirac--Fock model), but    removing the cutoff can then be a very delicate issue.  In fully relativistic external field QED one  inevitably has to face  ultra-violet divergences {originating} in singularities of Schwartz kernels of operators; see \cite{Brouder2002,Marecki2003,Indelicato2016,Deckert2023} for expository works focused on different formalisms and regimes.          

If the external potentials are smooth and possibly depend on time, singularities of propagators are best dealt with by  methods of microlocal analysis: these were initially applied in the  similar setting of Dirac fields on smooth Lorentzian spin manifolds \cite{Hollands2001a,Sahlmann2001} (cf.~\cite{Gerard2022a,Gerard2022,Drago2022,GHW,Islam2024,Hafner2024,Capoferri2025,Fewster2026} for more recent developments) and then reworked for external electromagnetic potentials \cite{Marecki2003,Wrochna2012a,Zahn2014,Zahn2015a,Galanda2025}. As in the case of scalar fields a primary role is played by the Hadamard condition \cite{radzikowski1996micro}, which implies short-distance behavior suitable for renormalization of perturbative series \cite{Brunetti2000a,Hollands2001b}.

If the potential is singular at $r=0$ and one is given  a  Hadamard state away from $r=0$,  then the constructions are still valid for $r\neq 0$ and in principle one has some well-defined renormalized quantities and perturbative series.  However, there are  two main problems with this:
\begin{enumerate}
\item It is not clear if a Hadamard state exists at all (or more to the point, it is not known if the natural physical states occurring in this context are Hadamard).\smallskip
\item Renormalized quantities such as the charge density could behave at $r=0$ {badly enough}  so that the quantum charge is not locally finite.
\end{enumerate}
The latter eventuality would mean infinite charge accumulation near $r=0$, and {consequently} the back-reaction of external electromagnetic fields with Dirac fields would no longer be negligible, thus undermining the consistency of the theory. 

The case of singular potentials is of particular importance in view of their role in modeling atoms, yet it is mathematically much more difficult because of their ability to disrupt regularity properties of propagators.  We remark that there  has been significant recent progress in studying spectral properties of Dirac--Coulomb operators (see e.g.~\cite{Esteban2020,Fournais2020,Esteban2021a,Esteban2021,Derezinski2022,Dolbeault2023}), but the methods employed so far are not suitable for analyzing local properties of singular Schwartz kernels.
 
\subsection{Main results}\label{ss:main} In the present paper we resolve the two
issues (1)--(2) in the case of \emph{Coulomb-like potentials},
possibly with a time-dependent smooth part. Namely, we consider a
Dirac--Coulomb Hamiltonian on $\rr^{3}$, i.e.~an operator of the
form
\beq\label{eq:as1}
H(t)=\sum_{k=1}^{3} \alpha_k \big(i^{-1} \p_k + A_k(t)\big) - A_0(t) + m \beta,
\eeq
where $m\in\rr$, $A_k\in C^\infty (\rr^{1,3})\cap C_\c^\infty(\rr_t; L^\infty(\rr^3_{\bar{x}}))$ and the potential $A_0$ is
of the form 
\beq\label{eq:as2}
A_0(t)=\frac{\kappa}{\module{\bar x}}+A_{0,\infty}(t), \quad {A_{0,\infty}\in C^\infty (\rr^{1,3})\cap C_\c^\infty(\rr_t; L^\infty(\rr^3_{\bar{x}}))},
\eeq 
with $\kappa\neq 0$ a constant.  We will suppose in particular that
\beq\label{eq:as3}
\supp A_k \cup \supp A_{0,\infty} \subset \{t \in \open{t_-, t_+}\}
\eeq
for some $t_-<t_+$.
Above, $\beta,\alpha_k$ are the $4\times 4$ matrices
$$
\beta= \begin{pmatrix} I &  0 \\ 0 & -I \end{pmatrix}, \quad \alpha_k= \begin{pmatrix} 0 & \sigma_k \\  \sigma_k & 0 \end{pmatrix}, \quad k=1,2,3,
$$
 where $\sigma_k\in M_{2}(\cc)$ are the Pauli matrices. Throughout the paper we assume that $ |\kappa |<\12$, in which case $H(t)$ is well-known to be essentially self-adjoint in $L^2(\rr^3;\cc^4)$.
  The Dirac--Coulomb  equation reads $Du=0$, where we adopt the convention that
  $$
  D=i \p_t + H(t).
  $$
    The corresponding Schrödinger evolution operator (which is simply $e^{i (t-s) H }$ if $H(t)\equiv H$ is $t$-independent) is denoted by $U(t,s)$.
  
In this setting, quantum states and second quantized fields  are  obtained  from a \emph{density matrix} at some reference time, i.e.~an operator $\gamma\in B(L^2(\rr^3;\cc^4))$ s.t.~$0\leq \gamma\leq \one$. If $\gamma(s)$ is the density matrix at time $s$, then the density matrix at time $t$ is given by
$$
\gamma(t)=U(t,s) \gamma(s) U(s,t). 
$$

The most important example is when $H(t)\equiv H$ is $t$-independent. Then the spectral projection  $\gamma=\one_{\opencl{-\infty,0}}(H)$ defines the fundamental physical state, called the  \emph{Dirac--Coulomb vacuum}. More generally, if  $H(t)\equiv H(t_-)$ is $t$-independent  in the past and $H(t)\equiv H(t_+)$ in the future as assumed in \eqref{eq:as2}--\eqref{eq:as3}, there are two natural possibilities. One can either take the Dirac-Coulomb vacuum $\one_{\opencl{-\infty,0}}(H(t_-))$ in the past, and  propagate it to the future, obtaining the {density matrix} $\gamma_{\rm in}(t)$ of the so-called \emph{in}-vacuum, or take $\one_{\opencl{-\infty,0}}(H (t_{+}))$ in the future, obtaining the \emph{out}-vacuum. The  difference $\gamma_{\rm out}(t)-\gamma_{\rm in}(t)$ is used to compute the charge created by the intermediate $t$-dependent potential, in fact the trace density 
\beq\label{eq:charge1}
\rho_{\gamma_{\rm out}, \gamma_{\rm in}}(t,\bar x):=\Tr_{\cc^4}(\gamma_{\rm out}(t)-\gamma_{\rm in}(t))(\bar x,\bar x)
\eeq
is the relative quantum charge density at $(t,\bar x)\in\rr^{1,3}$. In the absence of a cut-off, $\gamma_{\rm in / out}(t)$ are not trace-class and the trace densities (integrands of the trace, or Schwartz kernels restricted to the diagonal) do not exist, so the only hope for \eqref{eq:charge1} to be finite is that the singularities cancel out.

To understand how these singularities evolve in time, it is advantageous to think in terms of space-time quantities (as in QFT on curved spacetimes). Namely, to $\gamma$ one associates the spacetime \emph{two-point functions} $\Lambda^+$ and $\Lambda^-$, i.e., the operators
$$
\big(\Lambda^- f\big)(t)=\int_{-\infty}^\infty U(t,t_0)\gamma(t_0)  U(t_0,s) f(s)ds, \ \ f\in \ccf(\rr^{1,3};\cc^4),
$$
and similarly for $\Lambda^+$, with $\gamma$ replaced by $\one -
\gamma$.  The name comes from the fact that in QFT terminology, the
space-time smeared field operators $\ccf(\rr^{1,3};\cc^4)\ni f\mapsto \hat\psi(f),\hat\psi^*(f)$
associated to $\gamma$ satisfy
$$
\bra \Omega, \hat\psi(f) \hat\psi^*(g) \Omega\ket_{\cH} = \bra f, \Lambda^+ g \ket_{L^2}, \quad   \bra \Omega, \hat\psi^*(f) \hat\psi(g) \Omega\ket_{\cH} = {\bra g, \Lambda^- f \ket_{L^2}},
$$
where $\cH$ is the Hilbert space and $\Omega\in \cH$ the vacuum vector in the second quantized theory, see \S\ref{s:appA1}. The Schwartz kernel of $\Lambda^\pm$ is necessarily singular, but we can hope it is as regular as possible. In the setting with smooth potentials this occurs when $\Lambda^\pm$ satisfy the \emph{Hadamard condition}. This condition can be formulated as follows: for each non-zero $(t_i,\bar x_i;\tau_i, \bar \xi_i)\in \rr^{1,3}\times \rr\times \rr^3$ with $\pm\tau_i>0$, $i=1,2$, there exists pseudodifferential operators $B_1, B_2\in \Psi^0$ elliptic at $(t_1,\bar x_1;\tau_1,\bar \xi_1)$, $(t_2,\bar x_2;\tau_2,\bar \xi_2)$ respectively, such that $B_1 \Lambda^\pm B_2^*$ is regularizing (i.e., has smooth Schwartz kernel). Thus $\Lambda^\pm$ behave like smoothing operators after restricting localized Fourier transforms in time to $\pm \tau_i>0$, and since they also satisfy $D\Lambda^\pm=0$ and $\Lambda^\pm D=0$, this can be traded for some regularity in the spatial variables.

In the singular setting, we adopt the point of view that the Hadamard condition in the above formulation is still adequate for $r\neq 0$. We say that a density matrix $\gamma$ is \emph{Hadamard} if the associated two-point functions $\Lambda^\pm$ are Hadamard for $r\neq 0$. Our first main result says that the \emph{in}-vacua (and in analogy, \emph{out}-vacua) are Hadamard in this sense.

\begin{theorem}\label{thm:main1} Let $H(t)$ be the Dirac--Coulomb Hamiltonian in an electromagnetic potential $A_0(t),\dots,A_3(t)$ satisfying \eqref{eq:as1}--\eqref{eq:as3} with $ |\kappa |<\12$, in particular the potential is time-independent for $t\leq t_-$.  Suppose $0\notin \sp(H(t_-))$. Then the  density matrix
\beq
\gamma_{\rm in}(t):=U(t,t_-) \one_{\opencl{-\infty,0}}(H (t_{-})) U(t_-,t), \quad t\in \rr,
\eeq
of the $\rm in$-vacuum is Hadamard. 
\end{theorem} 

Consequently, for $r\neq 0$ all local constructions in perturbative QFT based on  \emph{in}-vacua proceed \emph{exactly} as in the non-singular case. In particular, the results on QFT in external potentials summarized in \cite{Zahn2014} and the earlier work \cite{Marecki2003} (building on similar results for Dirac fields on globally hyperbolic spacetimes) apply verbatim to this situation, hence the following immediate corollary.

\begin{corollary}\label{cor:QFT}  For any polynomial interaction, the Epstein--Glaser--Brunetti--Fre\-den\-ha\-gen perturbative interacting fermionic QFT based on $\gamma_{\rm in}$ is well-defined for $r\neq 0$.  
\end{corollary}

As stated, this result still does not say much about the behavior of renormalized quantities at $r=0$. However, we get a more precise statement  at least for relative charge densities.  Our second main result says in fact that renormalizing by subtracting a reference Hadamard density matrix yields a well-defined charge that  behaves reasonably well at $r=0$.

\begin{theorem}\label{thm2} Suppose {$\gamma,\gamma_{\rm ref}\in B(L^2(\rr^3;\cc^4))$} are two regular Hadamard density matrices at some arbitrary time $t\in \rr$. Then the relative charge density
$$
\rho_{\gamma, \gamma_{\rm ref}}(t,\bar x):= \Tr_{\cc^4}(\gamma(t)-\gamma_{\rm ref}(t))(\bar x,\bar x)\in L^1_{\rm loc}(\rr^3)\cap C^\infty(\rr^3\setminus\{0\} )
$$
is well defined, behaves like $O(r^{-2-\epsilon})$ near $r=0$ and is smooth in $t$.
\end{theorem}

Here \emph{regular}  means that the associated two-point functions
$\Lambda^\pm,\Lambda^\pm_{\rm ref}$ (viewed as operators) preserve
conormal regularity (i.e., regularity as measured by acting with $r\p_r$ and angular derivatives, as opposed to $\p_r$ and angular derivatives); see Definition \ref{def:reg}.

We show the regularity condition is satisfied by the \emph{in}- and \emph{out}-vacua, and as they are Hadamard by Theorem \ref{thm:main1}, Theorem \ref{thm2} applies to this case. This implies that the total charge created by turning on and then turning off a smooth time-dependent potential  is finite in any bounded region. 

\subsection{Method of proof, structure of paper} The main ingredient
of the proof is the propagation of singularities theorem for solutions
of the Dirac--Coulomb equation $Du=0$, obtained by two of us in
\cite{Baskin2023}, and generalized here to potentials with
time-dependent smooth part. 

Propagation of singularities in singular
settings has a long history: we mention here the work of
Cheeger--Taylor \cite{Cheeger1982}   (on wave equations on product cones),
Melrose--Wunsch \cite{Melrose2004} (on more general manifolds with
conical singularities), Vasy \cite{vasy2008propagation} (manifolds with corners) and Melrose--Vasy--Wunsch \cite{melrose2008propagation} (edge singularities). The Dirac--Coulomb equation poses significant
challenges as compared to the scalar case (cf.~Hintz \cite{Hintz2022}
for a parameter-dependent singular elliptic framework and hyperbolic
results adapted in particular to the Dirac--Coulomb operator,
Baskin--Booth--Gell-Redman \cite{Baskin2023a} for results on the
asymptotic behavior of solutions with microlocal methods, and Hintz \cite{Hintz2024}
for broad existence, uniqueness and regularity results applying to the
Dirac--Coulomb setting). On the other hand the
generalization to time-dependent potentials and lower a priori
regularity requires then only relatively mild modifications of the
arguments in \cite{Baskin2023} which we explain in the proof of
Theorem \ref{propagation}. The theorem uses Melrose's $\b$-calculus,
which permits us to localize singularities arriving at $r=0$ in time,
and to microlocalize them in energy, but not to distinguish among
singularities arriving at or leaving from $r=0$ along rays at
different angles; this permits \emph{diffractive} effects.

Instead of dealing directly with Schwartz kernels of operators we use an operator $\b$-wavefront set and show an operator version of propagation of singularities in  Theorem \ref{opp}. Its use is then twofold. First, we apply it to prove that if a bi-solution is  Hadamard for $r\neq 0$, then this implies a bound on its $\b$-wavefront set at $r=0$. Second, we show that the corresponding $\b$-wavefront set refinement of the Hadamard condition propagates well. This allows us to propagate regularity of $\gamma_{\rm in}$ and $\gamma_{\rm out}$ from the respective regions where $H(t)$ is time-independent. 

Once we establish control of operator $\b$-wavefront sets of
space-time quantities, we can use arguments similar to other singular
settings \cite{holographic,gannotwrochna} combined with a new
criterion (Proposition \ref{prop:td}) to conclude Theorem
\ref{thm2}. An important prerequisite, however, is the regularity
property, that is, we need good enough mapping properties of two-point
functions and propagators from $H^{-1}(\rr^{1,3};\cc^4)$-based
$\b$-Sobolev spaces of arbitrary order to
$H^{1}(\rr^{1,3};\cc^4)$-based $\b$-Sobolev spaces. We show this to
boil down to the well-posedness of the Cauchy problem in $\b$-Sobolev
spaces of arbitrary order including a good control of negative order
regularity; that turns out however to be problematic (we note in
particular that the results in \cite{Hintz2022,Hintz2024} do not
appear to give the regularity discussed here). The issue is that in our setting, domains of powers of
$H(t)$ are $t$-dependent in a way that makes them of limited use in
comparison to the $t$-independent setting,  and their interaction with
$\b$-Sobolev space analysis arguments turns out to be subtle. Our
strategy is to base the well-posedness theory on Cauchy data in the
static region where the use of domains of powers is sound, and to
derive energy estimates  smeared in time in the non-static region
rather than emphasizing Cauchy data there.

The paper is organized as follows. 

 Section \ref{sec2} provides the key results on regularity of solutions (in \S\ref{sec:propagator}) and propagation of singularities  in the case of $t$-dependent potentials (in \S\ref{ss:prop}), also in the operator sense (in \S\ref{ss:op}). Section \ref{sec3} proves the main results on Hadamard two-point functions and gives more precise estimates on their $\b$-wavefront sets. Finally, Appendix \ref{s:appA} recalls preliminaries on second quantization for Dirac fields needed for the formulation of our results and discusses corollaries for the renormalized quantum current and charge operators.

\section{Regularity of solutions and propagation of singularities} \label{sec2}

\subsection{Radial coordinates and blowup} In polar coordinates
$(r,\theta)$, the singularity at $\{r=0\}\subset \rr^3$ is dealt with
by replacing $\rr^3$ by the blown-up space 
$$
X=[\rr^3;\{0 \}]=\clopen{0,\infty}_r \times \ss_\theta^2
$$
equipped with the \emph{blowdown map}
$b: X\to \rr^3$ given by $b:(r,\theta)\mapsto r \theta$. We work on the blown-up spacetime $M=\rr_t \times X$, which can also be identified with the blow-up $[\rr^{1,3};\rr\times \{0\}]$.

Then, $X$ and $M$ are manifolds with boundary, and we denote by $X^\inti=\open{0,\infty}_r \times \ss_\theta^2$ and $M^{\inti}=\rr_t\times X^\inti$ their respective {interiors}.  

\subsection{Dirac--Coulomb propagator} \label{sec:propagator}
Throughout this section and the rest of the paper, {we write $H^1(M;\cc^4)$ for the pull-back under the blowdown map of the standard first order Sobolev space $H^1(\rr^{1,3};\cc^4)$, and denote its dual by $H^{-1}(M;\cc^4)$.}  The corresponding spaces with conormality of order {$s\in\rr\cup\{+\infty,-\infty\}$} are denoted by $H_\b^{1,s}(M;\cc^4)$ and $H_\b^{-1,s}(M;\cc^4)$; see \cite[\S3.1]{Baskin2023} for the definition.  In brief, $s$ stands for regularity with respect to taking $r\p_r$, $\p_t$, and $\p_\theta$    derivatives.  Note that conormality is always measured with respect to  $H^1(M;\cc^4)$  or its dual, so in particular $H_\b^{1,0}(M;\cc^4)$ equals $H^{1}(M;\cc^4)$ (rather than $H^{1}_\b(M;\cc^4)$, as the notation could perhaps wrongly suggest). In addition we add a subscript $H_{\b,\c}^{k,s}(M;\cc^4)$, resp.~$H_{\b,\loc}^{k,s}(M;\cc^4)$, to indicate the corresponding compactly supported and  local spaces, and often write in short $H_{\b,\c}^{k,s}$ and $H_{\b,\loc}^{k,s}$   when this creates no risk of confusion (and we abbreviate other spaces similarly).

For each $t\in \rr$, $H(t)$ is essentially self-adjoint on
$\cf_\c(X^\inti;\cc^4)$ in $L^2(X;\cc^4)=L^2(\rr^3;\cc^4)$.  As it
differs from the Dirac--Coulomb Hamiltonian by a bounded smooth
potential, its domain agrees with Kato's
characterization~\cite{Kato} of the domain of the Dirac--Coulomb
Hamiltonian.  In particular, {writing $\cD:=\Dom H(t)$, the space $\cD$ is independent of
$t$ and is given by} $H^{1}(X; \cc^{4})$, the closure of
$\cf_{\c}(X^{\inti} ; \cc^4)$ with respect to
$\norm{b_{*}\,\cdot}_{H^{1}}$, where $\norm{\cdot}_{H^{1}}$ is the
usual $H^{1}(\rr^{3};\cc^{4})$ Sobolev norm. {This is a global statement, including near spatial infinity. The graph norm on $\cD$ is equivalent to the $H^1(X;\cc^4)$ norm; near $r=0$, \cite[Lem.~7]{Baskin2023} identifies $H^1(X;\cc^4)$ locally with $r^{-1/2}H^1_\b(X;\cc^4)$.}  

We use the notation $\bra\lambda\ket=(1+\module{\lambda}^2)^{\frac{1}{2}}$ and define for $s\in\rr$
$$
\cD^s(t)= \Dom \bra H(t) \ket^s, \quad \cD^\infty(t)=\textstyle\bigcap_{s\in\rr} \cD^s(t),  \quad \cD^{-\infty}(t)=\textstyle\bigcup_{s\in\rr} \cD^s(t),
$$
where $\cD^s(t)$ is  equipped with the Hilbert space norm $\| \cdot \|_{\cD^s(t)}=\| \bra H(t)\ket^s \cdot \|$, and then {$\cD^\infty(t)$ is a Fréchet space with its usual projective limit topology, whereas $\cD^{-\infty}(t)$ is equipped with its usual locally convex inductive limit topology.}  These spaces are, in general,
time-dependent.  For all $t$, \cite[Lem.~14]{Baskin2023} implies that
$$
\cD^s(t)\cap \cE'(X^\inti;\cc^4) = H^s_\loc (X;\cc^4) \cap \cE'(X^\inti;\cc^4). 
$$
Because $\cD^{s}(t)$ regularity with large $s$ encodes partial power series
expansions of functions at $r=0$, though, the spaces
$\cD^{s}(t)$ at different times do not generally agree near the pole
of the potential.  We therefore
employ these spaces primarily at times $t \leq t_{-}$ for which $A_{k},
A_{0, \infty} \equiv 0$.  For convenience, we will drop the $t$ from the
$\cD^s(t)$ notation when a time has been fixed; we also omit it (as we have
done above) from
$\cD\equiv \cD^1(t)$, which is time independent.

If $A_{0,\infty}=A_{j}=0$, so that $H(t)$ is time-independent, the spectral
theorem implies that there is a unique solution of the Cauchy problem for the
homogeneous Dirac equation with
initial data in any of the spaces $\cD^s$.  For time-varying operators, however, this method is
unavailable, and we need more flexible tools (although it is useful
that solutions
via the spectral theorem still exist in the range of times when $A_{0,\infty}=0$).
 
\begin{lemma}\label{lem:Cauchy} For each $t_0\in \rr$, consider
    the Cauchy problem
$$
\begin{cases}
Du=f, \\ u|_{t=t_0}=u_0.
\end{cases}
$$

If $f \in C^{0}(\rr ; L^{2})$ and $u_{0} \in L^{2}$, then
 there is a unique solution $u \in C^{0}(\rr ; L^{2})$.  In addition, if $f \in C^{1}(\rr ; L^{2})$ and $u_{0}\in \cD$, then $u
  \in C^{1}(\rr ; L^{2}) \cap C^{0}(\rr ; \cD)$.  In both
  cases, $u$ is given by
  \begin{equation}
    \label{eq:formut}
    u(t) = U(t,t_{0})u_{0} {- i}\int_{t_{0}}^{t}U(t,t_{2})f(t_{2})\,dt_{2},
  \end{equation}
where $U(t_{1},t_{2})\in B(L^{2})$ satisfies for all $t, t_{1},
t_{2}\in \rr$:
\beq
\bea
&U(t,t) = \one, \quad U(t_1, t)U(t,t_2) = U(t_1, t_2), \\
& \p_t U(t,t_2) = i H(t)U(t,t_2), \quad \p_t U(t_1,t) = - U(t_1,t)iH(t),
\eea
\eeq
as a $\cD$-preserving strongly continuously differentiable family in
$B(\cD, L^{2})$.
\end{lemma}

\proof This
can be shown to follow from general semigroup theory; in particular it
suffices to check the hypotheses of \cite[\S 5,~Thm.~3.1]{pazy}, \cite[\S 5,~Thm.~4.8]{pazy} and
\cite[\S 5,~Thm.~5.3]{pazy}. For each $t\in \rr$, $H(t)$ is
self-adjoint on the $t$-independent domains $\cD$ in the Hilbert
space $L^{2}$. It follows that $\{ {i}H(t)\}_{t\in \rr}$ is a stable family
of infinitesimal generators of a $C^0$-semigroup on $L^{2}$ in the
terminology of \cite[\S 5, Def.~2.1]{pazy} (one can easily check it
directly, or invoke the more general \cite[\S
5,~Thm.~2.4]{pazy}). Furthermore, for all $u\in \cD$,
$H(\cdot)u$ is continuously differentiable in $L^{2}$, so all the
hypotheses are met.  \qeds

We now record a useful energy estimate which applies a priori to
solutions with higher regularity in $t$:
\begin{lemma}
  \label{lemma:energy-est}
Suppose $u \in {C}^k(\RR; \cD)$.   For $[t_{1}, t_{2}]\subset \rr$,  $k \in \nn$, and $i = 1$ or $i=2$, we have the uniform estimate
  \begin{equation*}
    \normm{\p_{t}^{k}u}_{L^{\infty}([t_{1},t_{2}] ; L^{2})} \lesssim
    \sum_{j\leq k}\left(\normm{\p_{t}^{j} Du}_{L^{1}([t_{1},t_{2}] ; L^{2})} +
    \normm{\p_{t}^{j}u(t_{i}, \bullet)}_{L^{2}}\right).
  \end{equation*}
\end{lemma}

\begin{proof}
  Note that for $k=0$, the statement follows from
  Lemma~\ref{lem:Cauchy} above.
  
  Since ${\p_{t}u} = -i Du + iHu$, we compute for $k\in \nn$, recalling
  that $H(t)$ is self-adjoint for each $t$:
  \begin{align*}
    \p_{t}\normm{\p_{t}^{k}u}^{2} &= 2 \Im \ang{ \p_{t}^{k}H(t)u,
                                   \p_{t}^{k}u} - 2 \Im
                                   \ang{\p_{t}^{k}Du, \p_{t}^{k}u} \\
    &= 2 \Im \ang{[\p_{t}^{k},H(t)]u, \p_{t}^{k}u} - 2 \Im
      \ang{\p_{t}^{k}Du, \p_{t}^{k}u}.
  \end{align*}
  Now $[\p_{t}^{k}, H(t)] = [\p_{t}^{k},-A_{0,\infty}] +
  \sum_{\ell=1}^{3}\alpha_{\ell}[\p_{t}^{k} , A_{\ell}]$ and so
  $[\p_{t}^{k}, H(t)] = \sum_{j < k}C_{j}\p_{t}^{j}$, where $C_{j}$
  are {uniformly bounded} matrix-valued smooth functions.  We thus get the bound
  \begin{align*}
    \p_{t}\normm{\p_{t}^{k}u}^{2}\lesssim
    \sum_{j=0}^{k-1}\normm{C_{j}\p_{t}^{j}u}\normm{\p_{t}^{k}u} + \normm{\p_{t}^{k}Du}\normm{\p_{t}^{k}u}
  \end{align*}
  and hence
  \begin{equation*}
    \p_{t}\normm{\p_{t}^{k}u} \lesssim
    \sum_{j=0}^{k-1}\normm{\p_{t}^{j}u} + \normm{\p_{t}^{k}Du}.
  \end{equation*}
  Integrating from $t_{i}$ and using induction then finishes the
  proof. \qeds
\end{proof}

We now fix $t_{-}< t_{-}'< t_{-}''$ so that $H(t)$ is $t$-independent for $t
\leq t_{-}''$.  We define the homogeneous solution operator $T$ by
\begin{equation*}
  Tu_{0} = U(t,t_{-})u_{0},
\end{equation*}
where $U(t,s)$ is as in Lemma~\ref{lem:Cauchy}.  In other words,
$Tu_{0}$ is the solution of the Cauchy problem posed at $t_{-}$ with
initial data $u_{0}$.  By Lemma~\ref{lem:Cauchy} we know $T: L^{2} \to
C^{0}_{\loc}(\rr ; L^{2})$ and $T: \cD \to C^{0}_{\loc}(\rr ;
\cD)\cap C^{1}_{\loc}(\rr ; L^{2})$.

The rest of this section is devoted to establishing that
$T : \cD ^{s}(t_-)\to H^{1,s-1}_{\b,\loc}$ for all $s \in \rr$.  We start
with the following elliptic lemma, which states that
$\b$-regularity of a solution can be verified by controlling its time
derivatives
(cf.~\cite[{Corollary~33}]{Baskin2023}). 
\begin{lemma}
  \label{lemma:elliptic-estimate}
  Suppose $k \in \zz$.  If $u \in H^{k}_{\loc}(\rr ; \cD)$ and $Du \in
  {H^{0,k}_{\b,\loc}}$, then $u \in {H^{1,k}_{\b,\loc}}$.
\end{lemma}
 {(We remark for context {that}
  for $k \in \mathbb{N}$, $H_{\b,\loc}^{1,k} \subset
  H_{\loc}^{k}(\rr; \cD)$, but the reverse inclusion does not hold unless
  we restrict our attention to solutions of the Dirac equation, as the
  functions in the space on the left enjoy higher spatial regularity.)}

\begin{proof}
  We begin by noting that $u \in H^{1,-\infty}_{\b,\loc}$.  Indeed,
  for $k > 0$, this follows immediately from the inclusion
  $H^{k}_{\loc}(\rr ; \cD) \subset H^{1}_{\loc}$, while for $k=0$, we note
  that $\p_{t}u = -iDu + iH(t)u \in L^{2}$ and so $u \in H^{1}$.  For
  $k < 0$, {by a standard local representation of negative Sobolev
  distributions, we can write $u = \sum_{j=0}^{|k|+1}\p_{t}^{j}v_{j}$ with
  $v_{j}\in H^{1}_{\loc}(\rr ; H^{1})\subset H^{1}_{\loc}$.  Thus
  $u \in H^{1,k-1}_{\b,\loc}\subset H^{1,-\infty}_{\b,\loc}$.}

  As the two notions of Sobolev regularity agree away from the
  singularity, it suffices to work near $r=0$.  In the
  time-independent case when {$k>0$}, the statement is an immediate
  consequence of Corollary~33 of~\cite{Baskin2023}.\footnote{There is
    a mistake in the statement of Corollary 33 in~\cite{Baskin2023}.
    The term $\widetilde{G}\eth u$ on the right side of the inequality
    should be estimated in $L^{2}$ rather than in $H^{1}$.}  The main
  changes needed in this setting are to extend the results to
  time-dependent coefficients and to relax the assumption of
  $u \in H^{1}$ to $u \in H^{1,-\infty}_{\b,\loc}$.  We employ the
  $\b$-pseudodifferential calculus, as described
  in~\cite[\S3]{Baskin2023} and the references therein.

  Given $u \in H^{1,-m}_{\b,\loc}$, we set $v = \Upsilon u$, where
  $\Upsilon \in \Psi^{-m}_{\b}(M)$ is an invertible, $SO(3)$-invariant,
  elliptic $\b$-pseudodifferential operator of order $-m$.  As
  $\Upsilon$ is invertible and elliptic, measuring $v$ in $H^{1,k+m}_{\b,\loc}$ is
  equivalent to measuring $u$ in $H^{1,k}_{\b,\loc}$ near $r=0$.
  The equation satisfied by $v$ is then
  \begin{equation*}
    \widetilde{D}v = \left( D - [D,\Upsilon]\Upsilon^{-1}\right)v =
    \Upsilon Du.
  \end{equation*}
  Using now an abbreviated notation for $\Psi_{\b}^{k}$, a computation
  in the differential $\b$-calculus shows that
  \begin{equation*}
    \widetilde{D}^{2} -D^{2} \in \Psi^{1}_{\b} +
    \frac{1}{r}\Psi^{0}_{\b} + \frac{1}{r^{2}}\Psi^{-1}_{\b} +
      \Psi^{-1}_{\b} \p_{r}^{2} + \Psi^{-1}_{\b}\frac{1}{r}\p_{r} + \Psi^{-1}_{\b}\frac{1}{r^{2}}\Delta_{\theta}.
  \end{equation*}

  The elliptic estimate~\cite[Lemma~28, Corollary~33]{Baskin2023} is
  based on the real part of the quadratic form $\ang{D^{2}u, u}$ and
  in fact holds for a wide class of Klein--Gordon-type operators.
  Inspection of the proof reveals that all of the
  components~\cite[\S5.2.2]{Baskin2023} of the elliptic estimate hold
  for the time-dependent operator $\widetilde{D}^{2}$; the key is
  Lemma~26, which still holds exactly as stated.  The remaining
  estimates in that section treat the remainder terms using only their
  order, which is unchanged here.  This yields an estimate for $v$ in
  $H^{1,k+m}_{\b,\loc}$ and thus an estimate for $u$ in
  $H^{1,k}_{\b,\loc}$, as desired.
  \qed
\end{proof}

\begin{lemma}\label{lemma:positive}
  For $s \geq 0$, $T: \cD^{s}(t_-) \to H^{1,s-1}_{\b, \loc}$.
\end{lemma}

\begin{proof}
  We first consider $s = k\in \nn$.  If $u_{0}\in \cD^{k}(t_-)$, then we
  set $v = e^{i(t-t_-)H(t_{-})}u_{0}$ for $t \leq t_{-}$.  Because $H$
  is independent of $t$ in this region, we know $v \in C^{k}(L^{2})\cap
  C^{0}(\cD^{k}(t_-))$ there.  {For $k\geq 1$, we also have
  $v \in C^{k-1}(\cD)$ there.}  Moreover,
    Lemma~\ref{lemma:elliptic-estimate} immediately yields $v \in H_{\b,\loc}^{1,k-1}$.

Let $\chi (t) \in C^{\infty}(\rr)$ satisfy $\chi(t) \equiv 1$ for
  $t\leq t_{-}$ and $\chi(t) \equiv 0$ for $t > t_{-}'$.  We seek a
  solution of the form $u = \chi v + w$ to $Du=0$, i.e., a solution of
  \begin{equation}\label{7.16:1}
    Dw =  {-[i\p_{t}, \chi]v}\in H_{\b,\loc}^{1,k-1}
  \end{equation}
  with zero initial data.  We can of course find an $L^2_\loc L^2$ (indeed,
  ${C}^0L^2$) solution $w$
  by the above semigroup result, Lemma~\ref{lem:Cauchy}; this is all
  we ask if $k=0$ (as the solution also lies in $H^{-1}_{\loc}
    (\rr; \cD)$), so w.l.o.g.~we can now take $k \geq 1$.
  We wish to use the energy estimate, Lemma~\ref{lemma:energy-est}, to
  obtain higher regularity on $w$, and to this
  end we regularize $w$ as follows.  Let $\Lambda_\ep$ be a family of
  translation invariant, properly-supported pseudodifferential
  operators of order $-\infty$ on $\RR_t$, bounded in $\Psi^0(\RR)$,
  and  with $\Lambda_\ep \to \one$ strongly on $H^s$ for all $s$ as $\ep\to 0$.  Then $\Lambda_\ep w
  \in {C}^\infty H^1$, and
\begin{equation}\label{smoothed}
D \Lambda_\ep w = [D,\Lambda_\ep] w{-\Lambda_\ep [i\p_{t}, \chi]v}.
\end{equation}
Note that since only the zeroth-order parts of $D$ are time-dependent, the
operators $[D,\Lambda_\ep]$ are bounded in $\Psi^{-1}(\RR)$ (and
still smoothing for $\ep>0$).

Now we apply the energy estimate to \eqref{smoothed}, which we may do
since both $\Lambda_\ep w$ and the right-hand side are smooth in $t$.  We know that $w$ is in
$L_\loc^2L^2$.  Inductively assuming $w \in H_\loc^{j-1} L^2$ (with $j\leq k$),
Lemma~\ref{lemma:energy-est} yields
$$
\normm{\pa^j_t \Lambda_\ep w}_{L^2_\loc L^2}\lesssim \norm{[D,\Lambda_\ep]
  w}_{H_\loc^j L^2} + \norm{\Lambda_\ep [i\p_{t}, \chi]v}_{H_\loc^j L^2}.
$$
As $\ep \to 0$, both terms on the RHS are bounded by the inductive
assumption and the assumption on $v$, so $\norm{\Lambda_\ep
  w}_{H_\loc^j L^2}$ is uniformly bounded as
$\ep \to 0$, hence has a weak sequential limit in $H_\loc^j L^2$.
Since $\Lambda_\ep w \to w$ in $L_\loc^2$, we have inductively obtained $w
\in H_\loc^k L^2$.  As
$H(t)w = -i \p_{t}w{- [i\pa_t, \chi] v} \in H^{k-1}_{\loc}L^{2}$, we also obtain
$$
w \in H^{k-1}_{\loc}(\rr; \cD).
$$
Now since \eqref{7.16:1} gave
$$
D w \in H_{\b,\loc}^{1,k-1} \subset H_{\b,\loc}^{0,k-1},
$$
{Lemma~\ref{lemma:elliptic-estimate} yields $w \in H_{\b,\loc}^{1,k-1}$, and hence $u=\chi v+w \in H_{\b,\loc}^{1,k-1}$, as desired.}  Interpolation then finishes the proof
  for general $s \geq 0$. {\qeds}
\end{proof}

The proof of Lemma~\ref{lemma:positive} also establishes the following
consequence for the inhomogeneous problem, which we state in
time-reversed form in anticipation of our later application of it:
\begin{lemma}
  \label{lem:inhomog-unique}
  If $k \geq 0$ and $f \in H^{k}_{\loc}L^{2}$ is supported in $t \leq T$, then there
  exists $v \in H^{k}_{\loc}L^{2}$ with $Dv = f$ and $v = 0$ for $t \geq T$.
\end{lemma}

Proving the analogue of Lemma~\ref{lemma:positive} for  negative $s$ requires a bit more
functional analysis.  For a fixed $k \in \nn$, we let $T \gg 0$ and set
\begin{equation*}
  \Yspace = \{ \phi \in \overline{H}^{k}(\clopens{t_{-}', \infty}; L^{2}(X))
  \mid \supp \phi \subset {[t_{-}', T]}\},
\end{equation*}
so that $\Yspace$ consists of extendible (in $t$) distributions at
$t_{-}'$ and supported distributions at $T$.  {We equip $\Yspace$ with the norm induced by
$\overline{H}^{k}(\clopens{t_{-}', \infty}; L^{2}(X))$; then $\cY$ is a Hilbert space as it is closed  subspace of $\overline{H}^{k}(\clopens{t_{-}', \infty}; L^{2}(X))$.}

The following lemma follows immediately from the energy
estimate (Lemma \ref{lemma:energy-est}) with $t_{2} = T$:
\begin{lemma}
  \label{lem:energy-est-renamed}
  For all test functions $\phi \in {C^{\infty}_{\rm c}(\clopens{t_-',\infty}\times X^{\inti})}\cap \Yspace$,
  \begin{equation*}
    \norm{\phi}_{\Yspace} \lesssim \norm{D^{*}\phi}_{{\Yspace}}.
  \end{equation*}
\end{lemma}

\begin{lemma}\label{lemma:negativeorders}
  For $k \in \nn$, $T: \cD^{-k} \to H^{-k-1}_{\loc}(\rr; \cD)$.
\end{lemma}

\begin{proof}
  Let $f \in \Yspace^{*}$.  For all test functions $\phi \in
  {C^{\infty}_{\rm c}(\clopens{t_-',T}\times X^{\inti})}\cap \Yspace$, we define
  \begin{equation*}
    F(D^{*}\phi) = \ang{\phi , f}.
  \end{equation*}
  By Cauchy--Schwarz and Lemma~\ref{lem:energy-est-renamed}, $F$ is
  well-defined and enjoys the estimate
  \begin{equation*}
    \abs{F(D^{*}\phi)} \lesssim \norm{D^{*}\phi}_{{\Yspace}}\norm{f}_{\Yspace^{*}}.
  \end{equation*}
  We can extend $F$ from the range of $D^{*}$ on test functions to all
  of ${\Yspace}$ by Hahn--Banach; {denoting the extension by $u\in\Yspace^{*}$, we have}
  \begin{equation*}
    \ang{\phi, f} = F(D^{*}\phi) = \ang{D^{*}\phi, u}.
  \end{equation*}
  Thus $Du = f$ distributionally wherever $\phi$ is allowed to have
  support, i.e., for $t < T$ as $\phi \in \Yspace$.  Note that we have
  implicitly obtained weak vanishing of $u$ at the left endpoint
  $t_{-}'$ because (locally) $\Yspace^{*}$ is in the dual space to a
  space of extendible distributions and hence is a space of supported
  distributions.  The equation \emph{is} satisfied down to the left
  endpoint $t=t_{-}'$ where $u$ vanishes, as the test functions may be
  supported there.

Thus given any $f \in \Yspace^{*}$ there is a solution $u\in
  {\Yspace^{*}}$ to $Du = f$.  We claim this solution is unique.  To see
  this, suppose $u \in {\Yspace^{*}}$ with $Du=0$.  For any test
  function $\psi$ with compact support in $\clopens{t_-', T} \times X^\circ$,
  by Lemma~\ref{lem:inhomog-unique}, there exists $v \in H_{\loc}^kL^2$ with $Dv=\psi$ and $v=0$
  for $t\geq T$. Owing to the complementary support properties of
  $u,v$, with $u$ {vanishing} for $t \leq t_-'$ and $v$
  vanishing by construction for $t \geq T$, we may integrate by parts
  to find $$0=\ang{Du,v} =\ang{u,Dv}=\ang{u,\psi}.$$ Consequently, $u$ vanishes.

  Now fix $u_{0}\in \cD^{-k}(t_-)$.  Let $v = e^{i(t-t_-)H(t_{-})}u_{0}$ be the
  locally-defined solution near $t=t_{-}$ of $Dv=0$ with Cauchy data
  $u_{0}$ given by the spectral theorem.  From the spectral-theoretic
  description, we further have
  $v \in H^{-k}_{\loc}(\opens{-\infty, t_{-}''} ; L^{2})$; we aim to extend
  it to a global solution in $H^{-k}_{\loc}(\rr ; L^{2})$.  Let
  $\chi(t)$ be a smooth function supported in $\opens{-\infty, t_{-}''}$
  with $\chi \equiv 1$ on a neighborhood of $\opencls{-\infty, t_{-}'}$.  We
  seek a solution $u$ of the form $u = \chi v + w$, with $w$ supported
  in $\clopens{t_{-}', \infty}$; such a solution would automatically have the
  correct initial data as $u=v$ near $t=t_{-}$.  For $u$ to be a
  solution, we require
  \begin{equation*}
    Dw = - [D, \chi]v = -i\chi'(t) v \in H^{-k}_{\rm c}(\rr ; L^{2}),
  \end{equation*}
  with $\supp w \subset \clopens{t_{-}', \infty}$.  {Applying the preceding
  construction for arbitrarily large $T$ and using uniqueness on overlaps gives such a
  global solution,} so the map sending Cauchy data to solution extends to
  \begin{equation*}
    T: \cD^{-k}(t_-)\to H^{-k}_{\loc}(\rr ; L^{2}(X)).
  \end{equation*}
  Moreover, since $H(t) u = - i\p_{t}u \in H_{\loc}^{-k-1}(\rr ; L^{2}(X))$, the
  characterization of the domain of $H(t)$ gives  $u \in
  H_{\loc}^{-k-1}(\rr ; H^{1}(X))$, as claimed. \qeds
\end{proof}

\begin{proposition}\label{lemma:negative}
  For $s \geq 0$, $T: \cD^{-s}(t_-) \to H_{\b, \loc}^{1, -s-1}$.
\end{proposition}

\begin{proof}
  Given $k\in \nn$ and $u_{0}\in \cD^{-k}(t_-)$, let $u = Tu_{0}$.  Lemma~\ref{lemma:negativeorders}
  asserts that $u \in H^{-k-1}_{\loc}(\rr ; \cD)$ and Lemma~\ref{lemma:elliptic-estimate}
  that $u \in H_{\b, \loc}^{1,-k-1}$.  Interpolation in the
  $\b$-regularity index then finishes the proof.\qeds
\end{proof}

\subsection{Propagation of singularities and diffraction} \label{ss:prop}
For each $u\in H^{1,-\infty}_{\b,\loc}(M;\cc^4)$, resp.~$H^{0,-\infty}_{\b,\loc}$, we denote by $\wf_\b^{1,\infty}(u)\subset\be T^*M$, resp.~$\wf_\b^{0,\infty}(u)$ the wavefront set  which microlocalizes infinite order conormal regularity with respect to $H^{1}(M;\cc^4)$, respectively $H^{0}(M;\cc^4)$; see e.g.~\cite{Baskin2023} for a pedagogical introduction. In particular, 
$$
\wf_\b^{1,\infty}(u) =\emptyset \ \Longleftrightarrow \ u \in {H^{1,\infty}_{\b,\loc}(M)}.
$$

Let $\Sig=\{(x,t,\xi,\tau)\in T^*M \st   \tau^2-\xi^2=0 \}$ be the
characteristic set associated to the Minkowski metric. The
\emph{compressed characteristic set} $\dSig$  is the image of $\Sig$
under the natural map $T^*M\to \be T^*M$.
By a
\emph{bicharacteristic} we will mean here the image in
$\dSig\subset\be T^*M$ of a null bicharacteristic for the Minkowski
metric.

Let $\sigma,\tau,\eta$ be the fiber coordinates on $\be T^*M$ in which the canonical one-form is
\beq\label{coord}
\sigma\frac{dr}{r} + \tau dt + {\eta \cdot d\theta}.
\eeq
Then
$$
\dSig\cap{\{r=0\}}=\{ ( r=0, \theta, t , \sigma=0, \eta=0, \tau ) \st \theta\in \ss^2, \ \tau\neq 0 \}.
$$
{Let $\dSig_0$ denote the topological space obtained as the
  quotient of $\dSig$ by the equivalence relation defined in the
  coordinates $(r,\theta,t,\sigma,\eta,\tau)$ above by
  $$
(0,\theta,t,0,0,\tau)\sim (0,\theta',t,0,0,\tau),\quad \theta,\theta'
\in S^2,
$$
i.e., identifying the factor of $S^2$ in the base.}

\begin{definition}\label{def:db} A \emph{diffractive bicharacteristic} is either:
\ben
\item\label{def:db1} a bicharacteristic in $\be T^* M^\inti$,
\item\label{def:db2}  a  bicharacteristic starting or ending at $\{r=0\}$,
\item\label{def:db3}  {a concatenation of two forward, or
  two backward bicharacteristics, one  ending at $\{r=0\}$ and one
  starting at $\{r=0\}$, whose projection to $\dSig_0$ is continuous.}
 \een
\end{definition}

In the case \eqref{def:db1} this is simply a null bicharacteristic (in the usual sense) which does not meet $\{r=0\}$. Note that $\be T^*M^\inti$ is naturally identified with $T^* M^\inti$, and Hörmander's propagation of singularities along null bicharacteristics applies. {Indeed, on $M^\inti$ the Dirac operator can be composed with a first order differential operator obtained by reversing the sign of the mass, to give a principally scalar second-order wave operator.} By equivalence of the wavefront set and $\b$-wavefront set on $M^\inti$, this means in particular that if $Du=0$ and $q_1,q_2\in \dSig$ are connected by a bicharacteristic lying entirely in $\be T^* M^\inti$  then
$$
 q_1 \notin \wf_\b^{1,\infty}(u)   \implies  q_2 \notin \wf_\b^{1,\infty}(u).
$$

In thinking of the second and third cases, we can  distinguish between \emph{incoming} and \emph{outgoing}
diffractive bicharacteristics (meaning ones moving \emph{towards},
resp.~\emph{away from} $r=0$ as $t$ increases)   by observing that
they are contained in  
$$
\bea
\Sigin&\defeq\dSig\cap\{{\eta=0,\ }\sigma/\tau>0\text{ or }r=0\},\\
\Sigout&\defeq\dSig\cap\{{\eta=0,\ }\sigma/\tau<0\text{ or }r=0\}.
\eea
$$
{Indeed, along these bicharacteristics one has $dr/dt=-\sigma/(r\tau)$ for $r\neq0$, so the two signs are opposite.}
Note that this terminology refers to the time flow on Minkowski space: this agrees with moving along bicharacteristics in the positive energy component $\dSig^+= \dSig\cap\{ \tau >0 \}$, but is the opposite of moving along bicharacteristics in the negative energy component $\dSig^-= \dSig\cap \{\tau < 0\}$. 

 We use the notation $q=(r,\theta,t,\sigma,\eta,\tau)$,
 $q'=(r',\theta',t',\sigma',\eta',\tau')$ for points in $\be T^*M$
 (and $\be S^*M$), where the Greek letters indicate fiber coordinates
 given by  \eqref{coord}.  To take into account the diffraction in the
 propagation of singularities we  first introduce the following
 terminology, {which deals with sets invariant under the same
   $S^2$ quotient used above in passing from $\dSig$ to $\dSig_0$.}
  
 \bed
 For $q=(r,\theta,t,\sigma,0,\tau) \in {\dSig}$ we define its \emph{diffractive spread} as the set
 $$
 S_q \defeq \{  q' \in   \dSig \st r'=r, \ \theta'\in \ss^2,  \ t'=t, \ \sgn \tau'=\sgn \tau, \ \sgn \sigma'=\sgn \sigma,  \ \eta'=0 \}.
 $$
 For $q=(r,\theta,t,\sigma,\eta,\tau)$ with $\eta\neq 0$, we simply set $S_q\defeq\{q\}$. 
 \eed
 
If $\sigma=0$, then by writing $\sgn \sigma'=\sgn \sigma$ we simply mean $\sigma'=0$.
 
On the base manifold, $S_q$ is the orbit of $q$ under spatial rotations around $r=0$ (in the case $\eta=0$). In the fibers, $S_q$ consists  of incoming or outgoing directions (depending on the nature of $q$) which lie in the same component $\dSig^\pm$ as $q$.
 
 \begin{definition} For $q_1,q_2\in \be T^*M$, we write $q_1\dot\sim  q_2$ if $q_1,q_2\in \dSig$ and $q_1$ is connected with $q_2$ by a diffractive bicharacteristic. 
  \end{definition}
{The continuity enforced in the third part of
  Definition~\ref{def:db} ensures that neither time $t$ nor energy
  $\tau$ jump upon interaction with the spatial origin, but by
  contrast the
  relationship of the angular variables of the concatenated bicharacteristics
  entering and leaving $r=0$ is unconstrained.
Of particular interest for our applications is that the
  constancy of $\tau$ along diffractive bicharacteristics enforces the following:} if $q_1 \dot\sim q_2$, then $q_1$ and $q_2$ lie in the same component  $\{ \tau >0\}$  or $\{\tau<0\}$, i.e.~either
  $$
  (q_1,q_2)\in \dSig^+\times \dSig^+, \mbox{ or } (q_1,q_2)\in \dSig^-\times\dSig^-.
  $$

In \cite[{Thm.~22--23}]{Baskin2023} the first and third authors prove a
diffractive propagation of singularities theorem which gives in
particular the following statement, generalized here to the
time-dependent setting introduced in \eqref{eq:as1}--\eqref{eq:as2}.  
  
 \begin{theorem}\label{propagation}  Suppose $u\in H_{\rm b,loc}^{1,-\infty} (M;\cc^4)$ satisfies $Du=0$. Then, $\wf_\b^{1,\infty}(u)\subset\dSig$. Furthermore, for all $q_1,q_2\in \dSig$ such that $q_1\dot\sim q_2$,
 \beq \label{eq:prop}
 S_{q_1} \cap \wf_\b^{1,\infty}(u)=\emptyset   \implies  S_{q_2} \cap \wf_\b^{1,\infty}(u) =\emptyset.
\eeq
More generally, if $Du=f$ for some $f\in H_{\rm b,loc}^{0,-\infty}(M;\cc^4)$ then \eqref{eq:prop} holds true away from $\wf_\b^{0,\infty}(f)$.
\end{theorem}

\proof  This result follows from Theorems 22 and 23 of \cite{Baskin2023} in the
special case of time independent Hamiltonian $H(t)=H(0)$. The proof
in the case of time-dependent Hamiltonian (where we recall that the
time dependence is only through the smooth part of the vector
potential) requires only fairly minor changes, however; we now
describe these changes to the proofs in \cite{Baskin2023}.  Note that we will
employ the calculus of b-pseudodifferential operators, as described in
\cite[\S3]{Baskin2023} and references therein.

The main changes needed in the time dependent case are twofold.  Most obviously, one must check that in the presence of
time dependent coefficients, the positive
commutators employed in \cite{Baskin2023} remain positive modulo the
same types of error terms as before.  Additionally, though, the
main commutator arguments in \cite{Baskin2023} are all built around
the background regularity assumption $u \in H^1(M;\cc^4)$ (or more generally $u\in C(\rr;\cD^{-\infty})$).   We claim that a local regularity assumption is sufficient, namely,  if the theorem is true for $u\in H^{1}(M;\cc^4)$ then it is also true under the weaker assumption $u\in H_{\rm loc}^{1}(M;\cc^4)$. In fact, we can find $\chi\in \cf_{\rm c}(M)$  such that $\chi=1$ on the base projection of the diffractive bicharacteristic segments from $q_1$ to $q_2$. Then $\chi u \in H^1(M;\cc^4)$ and {$D(\chi u)=\chi f+[D,\chi]u$}, with {$\supp [D,\chi]u$} disjoint from the region of interest, so we are reduced to applying the $H^{1}(M;\cc^4)$ version of the theorem.

To extend to
the more general assumption $u \in H_{\rm b,loc}^{1,-\infty}(M;\cc^4)$ requires no great
effort in the time-independent case, as we can simply commute $H$ with
powers of $\p_t$ (and keep in mind that $\p_t^k \in \Psi_\b^k(M)$) in
order to shift regularity up and down on the scale of spaces $H_\b^{1,s}(M;\cc^4)$.
On the other hand,
if $H$ is time-dependent, this argument to shift regularity no longer
applies as written.  We thus need to use a more robust argument,
conjugating by an elliptic b-pseudodifferential operator which does
not, in general, commute with $H(t)$.  We begin by describing this conjugation.

Given $u \in \Hl{-m}$ with $Du =f \in H_{\b,\loc}^{0,-m+1}$, we will proceed by proving
propagation of singularities for $v = \Upsilon u$,
where $\Upsilon \in \Psi^{-m}_\b(M)$ is an invertible, $\operatorname{SO}(3)$-invariant, elliptic
$\b$-pseudodifferential operator of order $-m$.  As $\Upsilon$ is
elliptic, measuring
$\wf_{\b}^{1,s}(u)$ is equivalent to measuring $\wf_{\b}^{1,s+m}(v)$.
The equation satisfied by $v$ is then
\begin{equation*}
  \widetilde{D}v = \left( D - \left[ D,
      \Upsilon\right]\Upsilon^{-1}\right) v = \Upsilon f.
\end{equation*}
Establishing propagation of singularities for $\widetilde{D}$ at the
level of $\Hl{s}$ for $s\geq 0$ provides the statement for (possibly
time-dependent) $D$ at the level of $\Hl{s-m}$.  Thus, it
  suffices for us to prove propagation of singularities for the
  operator $\widetilde{D}$ for a solution lying in $H_\loc^1(M;\cc^4)$.

As in the proof of Lemma~\ref{lemma:elliptic-estimate}, a computation in the differential $\b$-calculus verifies that
\begin{equation*}
  \widetilde{D}-D \in \Psi_{\b}^{0} + \Psi_{\b}^{-1}\frac{1}{r} +
  \Psi_{\b}^{-1}\p_{r} + \Psi_{\b}^{-1}\left( \frac{1}{r}\beta K\right),
\end{equation*}
where $K$ is Dirac's $K$-operator, see \cite[\S2.2]{Baskin2023}. 

We now discuss the changes in the proof of the propagation of
singularities results \cite[Thm.~22--23]{Baskin2023} necessary
to accommodate both the time-varying coefficients and also the need to
consider $\widetilde{D}$ rather than $D$ itself.  The argument has two main steps: an elliptic estimate and a
hyperbolic estimate.  The elliptic estimate, discussed in the proof of
Lemma~\ref{lemma:elliptic-estimate}, is based on the real part
of the quadratic form $\langle Pu, u\rangle$ and is used in two main
ways, both to provide estimates outside the characteristic set and to
more generally reduce the problem of estimating an $H^{1}$ norm to one
of estimating only the $L^{2}$ norm of $\p_{t}v$.  The other main
component is the hyperbolic estimate, which is obtained via a positive
commutator estimate built out of a microlocalization of
$r\partial_{r}$.

Both pieces require controlling the interactions of $r^{-1}$ and
$\p_{r}$ with Melrose's $\b$-calculus and the proofs of both pieces
hold for $\widetilde{D}$ as well.  The elliptic estimate was treated
above in Lemma~\ref{lemma:elliptic-estimate}.

For the hyperbolic estimate, the core arguments are
unchanged.  {Lemma~26} of \cite{Baskin2023} computes the commutator of $D$ with
a $\b$-pseudodifferential operator of order $m$ and is the main change
that must be made; there is a new term of the form
$\Psi_{\b}^{m-2}r^{-1}\beta K$.  On the other hand, similar terms
arise already in the application to a specific commutator~\cite[Lem.~36]{Baskin2023}, where this term is represented as $\mathbf{R}_{3}$.
As this term is treated as an error term and not as a source of
positivity, it yields no change in the proof of the hyperbolic
estimate.  {Note that the time dependence of the coefficients, per se, contributes no new types of
error terms.}
 \qeds

Note also that if the bicharacteristic segment connecting $q_1$ and $q_2$ does not meet $r=0$, then of course Hörmander's propagation of singularities theorem applies and $S_{q_i}$ can be replaced by $q_i$, $i=1,2$.

\subsection{Operatorial $\b$-wavefront set}\label{ss:op} Let us now introduce an operatorial version of the $\b$-wavefront set and discuss its properties,  analogously to \cite{holographic,gannotwrochna} (in particular the proofs of Lemmas \ref{emptyWF} and  \ref{wfs} below are fully analogous). We first define a class of operators that are suitable for $\b$-wavefront set considerations   and propagation of singularities.

\begin{definition}\label{def:reg} We say that $\Lambda$ is \emph{regular} if
$\Lambda: \Hcd{m}\to\Hl{m}$ continuously for all $m\in \rr$.
\end{definition}

\begin{definition} For $\Lambda$ regular,  its \emph{operatorial $\b$-wavefront set} is the set $\wf'_\b(\Lambda)\subset \cotwice$ defined as follows: $(q_1,q_2)\notin \wf'_\b(\Lambda)$ iff there exist  $B_i\in \Psi_\b^0(M)$, elliptic at $q_i$ ($i=1,2$),  such that 
\beq\label{eqregularizes}
B_1 \Lambda B_2^*: \Hcminfd\to \Hlinf
\eeq
continuously.
\end{definition}

\begin{lemma}\label{emptyWF}
Suppose $\Lambda$ is regular and $\wf'_\b(\Lambda)=\emptyset$. Then $\Lambda:\Hcminfd\to \Hlinf$ continuously.
\end{lemma}

\begin{lemma}\label{wfs} If $\Lambda$ and  $\tilde\Lambda$ are regular  then $\wf_\b'(\Lambda+\tilde\Lambda)\subset \wf_\b'(\Lambda)\cup \wf_\b'(\tilde\Lambda)$.
\end{lemma}

\begin{definition} We say that $\Lambda$ is a \emph{regular bi-solution} if it is regular and
 \beq
 D\circ \Lambda=0 \mbox{ on } \Hcminfd, \quad  \Lambda \circ D=0 \mbox{ on } \dHcminf.
\eeq
\end{definition}

We will use the following operator version of Theorem  \ref{propagation}.

 \begin{theorem}\label{opp} Suppose that $\Lambda$ is a regular  bi-solution.   Then $\wf'_\b(\Lambda)\subset \dSig\times\dSig$. Furthermore, if $q_1\in \dSig$ and $\Gamma_2\subset \dSig$ is closed then   for all $q_1'$ s.t.~$q_1 \dot\sim q_1'$,
 \beq\label{eq:opp1}
  (S_{q_1}\times \Gamma_2)\cap \wf'_\b(\Lambda)=\emptyset \implies  (S_{q_1'}\times \Gamma_2)\cap \wf'_\b(\Lambda)=\emptyset.
  \eeq
 Similarly, if $q_2\in \dSig$ and $\Gamma_1\subset \dSig$ is closed then    for all  $q_2'$ s.t.~$q_2 \dot\sim q_2'$,
 \beq\label{eq:opp2}
  (\Gamma_1\times S_{q_2})\cap \wf'_\b(\Lambda)=\emptyset \implies  (\Gamma_1\times S_{q_2'}) \cap \wf'_\b(\Lambda)=\emptyset.
 \eeq
 \end{theorem}
\proof Let us  show  \eqref{eq:opp1}. As the statement is microlocal, without loss of generality we can suppose that the projection of $\Gamma_2$ to the base manifold is compact.  Suppose $(S_{q_1}\times \Gamma_2)\cap \wf'_\b(\Lambda)=\emptyset$. Then there exist compactly supported $B_1,B_2\in \Psi^0_\b(M)$ elliptic on respectively $S_{q_1}$, $\Gamma_2$, such that for any bounded $\cU\subset \Hcminfd$, $B_1 \Lambda B_2^* \cU$ is uniformly bounded in  $\Hlinf$. We now apply propagation of singularities to $\Lambda B_2^* v$ for all $v\in \cU$ (in the form of the uniform estimate obtained in the proof of Theorem \ref{propagation}). The hypotheses are verified because $\Lambda B_2^* v$ is in $\dHlminf$  using that  $\Lambda$ is regular. As a result, for each point of $S_{q_1'}$ we can find $B_1'\in \Psi^0_\b(M)$ elliptic at that point, such that $B_1' \Lambda B_2^* \cU$ is uniformly bounded in  $\Hlinf$, which shows   \eqref{eq:opp1}.

We now show  \eqref{eq:opp2}. If $(\Gamma_1 \times S_{q_2})\cap \wf'_\b(\Lambda)=\emptyset$ for some compact $\Gamma_1$, and all compactly supported $B_1,B_2\in \Psi^0_\b(M)$ elliptic on respectively, $\Gamma_1$, $S_{q_2}$, with sufficiently small microsupport,  $B_1 \Lambda B_2^* :  \Hcminfd\to \Hlinf$. By taking adjoints
\beq\label{eq:blb}
 B_2  \Lambda^* B_1^*  :  H_{\rm b,c}^{-1,-\infty}(M;\cc^4) \to {H_{\rm b,\loc}^{1,\infty}(M;\cc^4)}.
\eeq
We  apply propagation of singularities to $\Lambda^* B_1^* v$ for $v$ in a bounded  subset  $\cU\subset H_{\rm b,c}^{-1,-\infty}$ similarly as before,   and   deduce  the existence of $B_2'$ elliptic at any chosen point of $S_{q_2'}$  such that 
$$
B_2'  \Lambda^* B_1^* : H_{\rm b,c}^{-1,-\infty}(M;\cc^4) \to {H_{\rm b,\loc}^{1,\infty}(M;\cc^4)}.
$$
By taking adjoints,
$$
B_1 \Lambda (B_2^{'})^* : H_{\rm b,c}^{-1,-\infty}(M;\cc^4) \to {H_{\rm b,\loc}^{1,\infty}(M;\cc^4)}.
$$
This implies \eqref{eq:opp2}.
\qed

\section{Singularities of two-point functions}\label{sec3}

\subsection{Hadamard condition}\label{ss:had} Let us first recall one version of the Hadamard condition, formulated here as a regularity condition \emph{away from the $r=0$ singularity}  (in fact, we do not expect to control the direct  analogue of that condition on the whole of $\rr^{1,3}$). 

\bed\label{def:had} We say that a pair of bounded operators $\Lambda^\pm:\ccf(M^\inti;\cc^4)\to\cD'(M^\inti;\cc^4)$ is \emph{Hadamard} if they satisfy (the \emph{Hadamard condition})
\beq\label{had}
\wf'(\Lambda^\pm)\subset \Sig^\pm\times \Sig^\pm,
\eeq
where $\wf'$ is the primed wavefront restricted to $S^*M^\inti \times S^*M^\inti$ and $\Sigma^\pm=\Sigma\cap \{ \pm\tau>0\}$ are the two connected components of the characteristic set $\Sigma$. 
\eed

{For a bounded operator $\Lambda:\ccf(M^\inti;\cc^4)\to\cD'(M^\inti;\cc^4)$, $\wf'(\Lambda)\subset S^*M^\inti \times S^*M^\inti$ can be defined by saying that $(q_1,q_2)\notin \wf'(\Lambda)$ iff there exist  properly supported $B_i\in \Psi^0(M)$, elliptic at $q_i$ ($i=1,2$),  such that $B_1 \Lambda B_2^*$ is smoothing.}
On $S^*M^\inti \times S^*M^\inti \simeq (T^*M^\inti\setminus \zero)/\RR_+\times (T^*M^\inti\setminus \zero)/\RR_+$, this definition of $\wf'$ coincides with the standard definition, which refers to the  wavefront set of the Schwartz kernel, {see e.g.~the proof of \cite[Lem.~3.1]{Dang2020}}. In contrast, however, it does not say anything about possible singularities located at $\zero\times T^* M^\inti$ or $T^* M^\inti \times \zero$. In practice, this type of information is replaced by the regularity assumption $\Lambda:\ccf(M^\inti;\cc^4)\to\cD'(M^\inti;\cc^4)$.

We stress that in the presence of a singular potential, it is not immediately clear which formulation of the Hadamard condition is adequate, even away from the singularity.  Here, we use the Hadamard condition in the form \eqref{had}, following \cite{Sahlmann,Hollands2001a}.  If the potentials are smooth at $r=0$ as opposed to having a Coulomb singularity, and if $\Lambda^\pm$ are two-point functions on $\rr^{1,3}$ (or on a more general globally hyperbolic spacetime), then  condition \eqref{had} on $\rr^{1,3}$ is equivalent to the more precise (and most commonly used) statement
$$
\wf'(\Lambda^\pm) \subset   (\Sigma^\pm\times\Sigma^\pm)\cap \{ (q_1,q_2) \st q_1 \sim q_2 \},
$$
where $q_1\sim q_2$ means that $q_1$ is connected with $q_2$ by a bicharacteristic in the usual (non-diffractive) sense. This equivalence needs however not be true in the presence of a Coulomb singularity because of diffractive bicharacteristics that cross $r=0$. The correct replacement is given later on in Proposition \ref{preciseWF}.

In practice, to control the Hadamard condition \eqref{had} we  need to consider the operator $\b$-wavefront set on the whole of $M$.
 
 In what follows we will not distinguish between $\Sigma \cap \{ r\neq0\}$ and $\dSig\cap\{r\neq 0\}$ in the notation (recall that they are identified by a canonical isomorphism).
 
 \begin{proposition}\label{prop:stronghad}  If  $\Lambda^\pm$ is a pair of regular  bi-solutions, then it is Hadamard if and only if
 \beq\label{bhad}
 \wf'_\b(\Lambda^\pm)\subset \dSig^\pm \times \dSig^\pm.
 \eeq
 \end{proposition}
\proof For any $q\in \dSig^\mp\cap \{ r=0\}$, we can find $q'\in \Sig^\mp \cap \{r\neq 0\}$ such that $q'\dot\sim q$. Since $S_{q'}\subset \Sig^\mp$, the Hadamard condition implies
$$
(S_{q'}\times (\dSig\cap\{ r\neq 0 \}))\cap \wf'_\b(\Lambda^\pm)=\emptyset=  (\dSig\cap \{ r\neq 0 \}\times S_{q'})\cap \wf'_\b(\Lambda^\pm)
$$
By Theorem \ref{opp} (operator version of propagation of singularities), this yields
$$
(\{q\}\times (\dSig\cap  \{ r\neq 0 \}))\cap \wf'_\b(\Lambda^\pm)=\emptyset=  ((\dSig\cap  \{ r\neq 0 \})\times \{q\})\cap \wf'_\b(\Lambda^\pm)
$$
{Now let $(q_1,q_2)\in (\dot \Sigma^\mp \times \dot \Sigma)\cup(\dot \Sigma\times \dot \Sigma^\mp)$. For each $q_i$ lying over $r=0$, choose $q_i'\in\dot\Sigma\cap\{r\neq0\}$ such that $q_i'\dot\sim q_i$, and set $q_i'=q_i$ otherwise. Since the sign of $\tau$ is preserved along diffractive bicharacteristics, at least one of $q_1',q_2'$ lies in $\dot\Sigma^\mp$. The Hadamard condition \eqref{had} therefore excludes $S_{q_1'}\times S_{q_2'}$, and successive applications of Theorem \ref{opp} in the variables for which $q_i\neq q_i'$ give $(q_1,q_2)\notin\wf'_\b(\Lambda^\pm)$.} For points not in $\dot \Sigma \times \dot \Sigma$ it suffices to invoke the elliptic estimate statement in Theorem \ref{opp}. The opposite direction is evident.
\qed 
\smallskip

We will show that the differences of regular bi-solutions have a
well-defined on-diagonal restriction, and hence a well-defined trace
density. We start by recalling an elementary fact about mapping
properties of the Mellin transform $\cM$, defined here with the normalization
$$
(\cM f)(\xi) = \int_0^\infty f(r) r^{-i \xi - 1} dr.
$$

\bel \label{lem:M}
If $\cM f(\xi)$ is holomorphic and $O(\bra \xi\ket^{-s})$ in $\Im \xi >- \epsilon$ for some $\epsilon >0$ and $s>1$, then $f=O(1)$ near $r=0$, with
$$
\sup_{r\in \open{0,1}}\module{f(r)} \leq \frac{1}{2\pi} \int_{-\infty}^{\infty} | \cM f(\xi)| d\xi.
$$
\eel
\proof We write $f$ as the inverse Mellin transform
$$
f(r)= \frac{1}{2\pi}\int_{-\infty}^{\infty} (\cM f)(\xi) r^{i\xi} d\xi
$$
and bound the integral accordingly.
\qed

\begin{proposition}\label{cor1} Suppose $\Lambda^\pm$ and $\tilde\Lambda^\pm$ are regular  bi-solutions that satisfy the Hadamard condition, and 
\beq\label{lambdaaa}
\Lambda^{+}-\tilde\Lambda^{+}=-(\Lambda^{-}-\tilde\Lambda^{-}).
\eeq 
Then $\Lambda^\pm - \tilde\Lambda^\pm :   \Hcminfd\to \dHlinf$ continuously.
\end{proposition}
\proof  The $\b$-wave front set of the l.h.s.~of \eqref{lambdaaa} is contained in $\dSig^+\times  \dSig^+$, and the $\b$-wave front set of the RHS is contained in $\dSig^-\times  \dSig^-$, hence the two are disjoint. Thus, both sides of \eqref{lambdaaa} have in fact empty $\b$-wave front set, and thus 
\beq\label{eq:regu}
\Lambda^{\pm}-\tilde\Lambda^{\pm}:    \Hcminfd\to \dHlinf
\eeq
is bounded by Lemma \ref{emptyWF}. \qed

Next we show that for the existence of a trace density we need slightly weaker mapping properties than \eqref{eq:regu}

\begin{proposition}\label{prop:td} If $\Lambda:
  H_{\b,\c}^{0,-\infty}(M;\cc^4)\to \dHlinf$ continuously then
  $\Lambda$ has a locally integrable trace density. 
\end{proposition}
\proof Since the statement is local we can without loss of generality assume that $\Lambda:   H_{\b,\c}^{0,-\infty}\to \dHcinf$.  The $\b$-regularity means that for any $N$ and choices of vector fields $V_1,\dots,V_N$ with $V_j \in \pi_i^* \cV_\b(M)$,  $i=1,2$ ($\pi_i$ being the projection on the $i$-th component of $M\times M$), the operator $\Lambda_V$ with Schwartz kernel
$$
V_1 \dots V_N  {\Lambda(t_1,t_2,r_1,r_2,\theta_1,\theta_2)} 
$$
still has the property $\Lambda_V: L^2(M;\cc^4)\to H^{1}_{\rm loc}(M;\cc^4)$. Recall that the volume form in our coordinates is $r^2 dt\, dr \, d\theta$. Since $r^{s}\in L^2_{\rm loc}(M)$ for $\Re {s}>-3/2$ and $r^{s'}\in H_{\rm loc}^{-1}(M)$ for {$\Re s'>-5/2$}, if we choose any pair of spherical harmonics $Y_\bullet$ on $\ss^2$ we obtain
$$
\big| \bra \Lambda_V r^{s} Y_{lm}, r^{s'} Y_{l'm'} \ket \big|\leq C < \infty, \quad \Re s > - 3/2, \ \Re s' > - 5/2
$$
with the constant $C$ uniform in $s,s'$ provided we restrict to strictly smaller half-planes, and also uniform in the orders $m,l,m',l'$ of the spherical harmonics. Thus for $\Re s > -3/2$, $\Re s'>-5/2$, 
$$
\big\lvert \int  V_1 \dots V_N {\Lambda(t_1,t_2,r_1,r_2,\theta_1,\theta_2)}
r_1^{s} r_2^{s'} Y_{lm}(\theta_1) Y_{l'm'}(\theta_2) \, dt_1 \, dt_2
\, r_1^2 dr_1 \,  r_2^2 dr_2 \,d\theta_1 \, d\theta_2 \big \rvert \leq C.
$$
Now we interpret the integrals in $r_1, r_2$ as Mellin transforms in $r_1,r_2$ respectively, denoting these maps as $\cM_{(1)},\cM_{(2)}$.  Thus our estimate is
$$
\big\lvert \int (\cM_{(1)} \cM_{(2)}  V_1 \dots V_N \Lambda)({t_1,t_2,}is+3i, i s' + 3i,\theta_1,\theta_2)Y_{lm}(\theta_1) Y_{l'm'}(\theta_2) dt_1\,dt_2\, d\theta_1 \, d\theta_2 \big\rvert\leq C,
$$
again for $\Re s >-3/2$, $\Re s' >-5/2$, and the l.h.s.~is holomorphic
in $s,s'$ in these half-spaces. Moreover by taking $V_j$ to be vector
fields of the form $\p_{t_1},\p_{t_2},\p_{\theta_1},\p_{\theta_2}$,
and integrating by parts so that the $\theta_1,\theta_2$ derivatives
fall onto the spherical harmonics, the uniform boundedness now shows
that for all $N\in \nn$,  and $(t_1,t_2)$ ranging over a compact set,
$$
\big\lvert\int   (\cM_{(1)} \cM_{(2)} \Lambda)({t_1,t_2,}is+3i, i s' +
3i,\theta_1,\theta_2)Y_{lm}(\theta_1) Y_{l'm'}(\theta_2)  d\theta_1 \,
d\theta_2 \big\rvert\leq C_N \bra l \ket^{-N} \bra l' \ket^{-N}
$$ 
 (with $l(l+1)$ being the Laplace eigenvalue of $Y_{lm}$); here we
 have employed the estimate on the vector fields $\pa_{t_1}$,
 $\pa_{t_2}$ to obtain a (locally) uniform estimate in
 $t_1,t_2$. Further integration by parts  with the $\b$-vector fields
 $V_j$ chosen to be $r_1 \p_{r_1}$, $r_2 \p_{r_2}$ improves the
 estimate to  
\beq\label{eq:MM}
\bea
\big\lvert \int   (\cM_{(1)} \cM_{(2)} \Lambda)({t_1,t_2,}is+3i, i s' + 3i,\theta_1,\theta_2)Y_{lm}(\theta_1) Y_{l'm'}(\theta_2) d\theta_1 \, d\theta_2 \big\rvert \\ \leq C_N 
 \bra l \ket^{-N} \bra l' \ket^{-N}  \bra s \ket^{-N} \bra s' \ket^{-N},  \quad  \Re s > -3/2, \ \Re s'>-5/2.\eea
\eeq

Applying Lemma \ref{lem:M}  to the l.h.s.~of \eqref{eq:MM} (and
recalling that {$\cM(r^k f)(\xi)=\cM f(\xi+i k)$}) shows that since for
any $\delta>0$,
$$
{\int \big(\cM_{(1)} \cM_{(2)}(r^{3/2+\delta}_1 r_2^{1/2+\delta}\Lambda)\big)(t_1,t_2,\zeta,\zeta',\theta_1,\theta_2) Y_{lm}(\theta_1) Y_{l' m'}(\theta_2)\, d\theta_1\, d\theta_2}
$$
is holomorphic in $\Im \zeta, \Im \zeta'>-\delta$, and bounded by
$$
C_N \bra l \ket^{-N} \bra l'\ket^{-N} \bra \zeta\ket^{-N}\bra \zeta'\ket^{-N},
$$
then in fact 
$$
{\int  r^{3/2+\delta}_1 r_2^{1/2+\delta} \Lambda(t_1,t_2,r_1,r_2,\theta_1,\theta_2) Y_{lm}(\theta_1) Y_{l'm'}(\theta_2)\,d\theta_1\,d\theta_2=O(\bra l\ket^{-N}) O(\bra l'\ket^{-N})}
$$
for all $\delta>0$.
{The rapid decay in $l,l'$ makes the spherical-harmonic series absolutely and locally uniformly convergent, so the weighted kernel has a continuous representative.} It follows that
$$
\big| r_1^{3/2+\delta} r_2^{1/2+\delta} \Lambda(t_1,t_2,r_1,r_2,\theta_1,\theta_2)\big|\leq C.
$$
Restricting to the diagonal  gives
$$
\Lambda|_{\rm diag} (t,r,\theta) = O(r^{-2-2\delta}) \in L_{\rm loc}^1({\rm diag}; r^2\, dt\, dr\, d\theta)
$$
as asserted, provided we choose $2\delta<1$.
\qeds

We conclude that if $\Lambda^\pm$ are a Hadamard pair then they have a well-defined \emph{relative trace density} (relative to some other  Hadamard $\tilde\Lambda^\pm$). 

\subsection{Two-point functions}  

Let $t_0\in \rr$ be some reference time. In the terminology commonly used for fermions, a \emph{quasi-free state} is determined by its \emph{density matrix}\footnote{This is not to be confused with the $\gamma$ matrices that enter the definition of the Dirac operator, see \S\ref{s:appA1}.} at time $t_0$, i.e.~by a bounded operator {$\gamma(t_0)$} on $L^2(X;\cc^4)$ such that
\beq
0 \leq \gamma(t_0) \leq \one.
\eeq

Its \emph{space-time two-point functions} are the pair of operators $\Lambda^\pm$ given by 
\beq\label{dlg1}
\bea
\Lambda^-(t_1,t_2)&=U(t_1,t_0)\gamma(t_0)  U(t_0,t_2), \\  \Lambda^+(t_1,t_2)&=U(t_1,t_0)(\one -\gamma(t_0)) U(t_0,t_2),
\eea
\eeq 
where Schwartz kernel notation is used in the time variable only, i.e.
$$
\big(\Lambda^- f\big)(t)=\int_{-\infty}^\infty U(t,t_0)\gamma(t_0)  U(t_0,t_2) f(t_2)dt_2
$$
and similarly for $\Lambda^+$.
 They satisfy
\beq
\Lambda^\pm \geq 0, \quad \Lambda^+ + \Lambda^- = S, 
\eeq
where $S$ is the operator with time Schwartz kernel $U(t_1,t_2)$.

\begin{lemma}\label{Sregular} The operator $S$ is a regular bi-solution, i.e.~$S: \Hcd{s}\to\Hl{s}$ continuously for all $s\in \rr$.
\end{lemma}
\proof Let $T$ be as given in Section~\ref{sec:propagator} with $t_{-}
  = t_{0}$.  By Lemmas~\ref{lemma:positive} and~\ref{lemma:negative},
  \begin{equation}\label{higherreg}
    T : \cD^{s}\to H^{1,s-1}_{\b,\loc}, \quad s \in \rr,
  \end{equation}
  and, dually,
  \begin{equation}
    \label{higherreg2}
    T^{*}: H^{-1,-s+1}_{\b,\c} \to \cD^{-s}, \quad s\in \rr.
  \end{equation}
  Hence for all $s\in \rr$,
  \begin{equation*}
    S = TT^{*} : H_{\b,\c}^{-1,s+1} \to H^{1,s-1}_{\b,\loc}.
  \end{equation*}
  \qeds

In the case when $H(t)\equiv H$ is time-independent, of particular importance are \emph{stationary states}, i.e.~those for which $\gamma(t)\equiv\gamma=\chi(H)$ for some Borel function $\chi$ with values in $[0,1]$.  If $0\notin\sp(H)$, then the \emph{Dirac--Coulomb vacuum} is obtained from the density matrix $\gamma=\one_{\opencl{-\infty,0}}(H)$. It is also natural to consider excited states: this corresponds to taking $\one_{\opencl{-\infty,\lambda}}(H)$ for some $\lambda\notin \sp(H) $.

More generally, since $H(t)$ is time-independent for $t\leq t_-$ and
$t\geq t_+$ then the \emph{in-state}, resp.~\emph{out-state} are defined
by taking $\gamma(t_0)=\one_{\opencl{-\infty,0}}(H(t_0))$ for some reference
time $t_0\leq t_-$, resp.~$t_0\geq t_+$.

\bep\label{t1} Suppose $t_0\leq t_-$ or $t_0 \geq t_+$.  If $\gamma(t_0)=\chi_-({H(t_0)})$ for some $\chi_-\in S^0(\rr;[0,1])$  such that $\chi_-\equiv \one_{\opencl{-\infty,0}}$ outside of a neighborhood of $0$, then the corresponding pair $\Lambda^\pm$ defined in \eqref{dlg1} is a pair of Hadamard regular bi-solutions.
\eep
\begin{proof}
  We show regularity for $\Lambda^{+}$; the statement for $\Lambda^{-}
  = S - \Lambda^{+}$ then
  follows immediately by the previous lemma.  Let $T$ be the propagator described in Section~\ref{sec:propagator}.
  We write
  \begin{equation*}
    \Lambda ^{+} = T (\one-\gamma(t_{0})) T^{*},
  \end{equation*}
  where $T: \cD^{s}\to H^{1,s-1}_{\b,\loc}$ and $T^{*}:
  H^{-1,s+1}_{\b,\c} \to \cD^{s}$ for all $s\in \rr$.  Because
  $\gamma(t_{0})$ commutes with $H(t_{0})$, $\gamma(t_{0})$ is bounded
  $\cD^{s}\to \cD^{s}$ for all $s\in \rr$ and so $\Lambda^{+}$ is
  regular.

  That $\Lambda^{\pm}$ is a bi-solution follows from its representation
 in terms of the solution
  operator $T$.  We must therefore only show that $\Lambda^{\pm}$ is
  Hadamard.  {Let $I$ be a non-empty open interval  contained in the
  relevant static region, and let $\Lambda^{\pm}_{0}$ denote
  the restriction of $\Lambda^{\pm}$ to $I\times X^{\inti}$, viewed as an
  operator
  $\ccf(I\times X^{\inti};\cc^4)\to\cD'(I\times X^{\inti};\cc^4)$.
  By Theorem~\ref{opp},
  $\WF'_{\b}(\Lambda^{\pm})\subset\dot\Sigma\times\dot\Sigma$
  globally, in particular 
  $\WF'(\Lambda_{0}^{\pm})\subset\Sigma\times\Sigma$ on $(I\times X^\inti)\times( I\times X^\inti)$ where $\WF'$ is meant as in Definition \ref{def:had} and in the paragraph below.}

  {Set $H_0=H(t_0)$ and
  $\chi_+:=1-\chi_-$, and define a  pair of operators   $\Lambda_{\rm stat}^{\pm}$ on $M$ by
  $$
  \Lambda_{\rm stat}^{\pm}(t_1,t_2)
  :=e^{i(t_1-t_0)H_0}\chi_\pm(H_0) e^{i(t_0-t_2)H_0}
    =\int_{\rr}e^{i(t_1-t_2)\lambda}\chi_\pm(\lambda) \,d\mu_{H_0}(\lambda),
  $$
where $d\mu_{H_0}$ is the  spectral measure of $H_0$. 
 By the same arguments as for $\Lambda^\pm$, $\Lambda_{\rm stat}^{\pm}$ are regular bi-solutions for the static Dirac operator, and moreover, $\Lambda_{\rm stat}^{\pm}$ agrees with
  $\Lambda^{\pm}$ on $I\times I$. Choose
  $\tilde\chi_{\mp}\in S^0(\rr)$ so that
  $\tilde\chi_{\mp}(\tau)=0$ for $\pm\tau\gg0$,
  $\tilde\chi_{\mp}(\tau)=1$ for $\mp\tau\gg0$, and $ \tilde\chi_{\mp}\chi_{\pm}=0$. 
  Then  
  $$
  \tilde\chi_{\mp}(i^{-1} \p_t)\Lambda_{\rm stat}^{\pm}=0
  =\Lambda_{\rm stat}^{\pm}\tilde\chi_{\mp}(i^{-1} \p_t)
  $$}
  While
  $\tilde{\chi}_{\mp}(i^{-1}\p_{t})$ is not a pseudodifferential
  operator on $M^{\inti}$, it is possible to find a properly supported
  $A \in \Psi^{0}(M^\circ)$ so that $A\tilde{\chi}_{\mp}(i^{-1}\p_{t})$ and
  $\tilde{\chi}_{\mp}(i^{-1}\p_{t})A$ are pseudodifferential and
  elliptic at any chosen point of $\Sigma^{\mp}$. {Hence, restricting our attention to $I\times X^\inti$,}
  $\WF'(\Lambda^{\pm}_{0})\subset \Sigma^{\pm}\times \Sigma$ and
  $\WF'(\Lambda^{\pm}_{0})\subset \Sigma\times \Sigma^{\pm}$   over $(I \times X^{\inti})\times (I\times X^{\inti})$ since $\Lambda^\pm_0$ coincides with $\Lambda^\pm_{\rm stat}$ there.  Thus,
  $\WF'(\Lambda_{0}^{\pm})\subset \Sigma^{\pm}\times \Sigma^{\pm}$
  over $(I \times X^{\inti})\times (I\times X^{\inti})$.  By
  {Proposition~\ref{prop:stronghad} applied on $I\times X$, this implies
  $\WF_{\b}'(\Lambda^{\pm})\subset \dot\Sigma
  ^{\pm}\times\dot\Sigma^{\pm}$ over $(I\times X)\times (I\times X)$.}
  Finally, by operator propagation of singularities, i.e.,
  Theorem~\ref{opp}, applied to the first and then to the second set
  of spacetime variables, we conclude $\WF_{\b}'(\Lambda^{\pm})\subset
  \dot\Sigma^{\pm}\times \dot \Sigma^{\pm}$ everywhere.
  \qed
\end{proof}

\medskip

In particular this proves Theorem  \ref{thm:main1}. {The strategy of propagating the Hadamard property from a static region to a non-static one is broadly used in variants of the Fulling--Narcowich--Wald argument for constructing Hadamard states, see e.g. \cite{FNW,Wrochna2012a,Murro2021,Gerard2022,Moretti2023}. Note that in our setting a particular role is played by the regularity condition and Proposition \ref{prop:stronghad} which enable the use of  propagation of singularities in the form of   Theorem~\ref{opp}. }

We can also estimate the $\b$-wavefront set of Hadamard two-point functions a bit more precisely.

\bep\label{preciseWF} If $\Lambda^\pm$ are a pair of two-point functions and Hadamard bi-solutions then 
\beq\label{preciseWF2}
\wf'_\b( \Lambda^\pm)\subset  (\dSig^\pm\times\dSig^\pm)\cap \{ (q_1,q_2) \st q_1 \dot\sim q_2 \}.
\eeq
\eep
\proof  Since $\Lambda^++ \Lambda^-= S$, 
 using  Proposition \ref{prop:stronghad}   $\Lambda^+$ and $\Lambda^-$ have disjoint wavefront sets and we get 
\beq\label{eqlpm}
\wf'_\b( \Lambda^\pm )\subset (\dSig^\pm\times\dSig^\pm)\cap \wf'_\b(S).
\eeq
Recall that $S=S_+-S_-$ where $S_+$ is the retarded, and $S_-$ the advanced propagator. By the elliptic regularity statement in Theorem \ref{opp},
\beq\label{eq:wfbS}
\wf'_\b(S)\subset \big(\wf'_\b(S_+)\cup  \wf'_\b(S_-)\big) \cap (\dot\Sigma\times \dot\Sigma).
\eeq
Now let $q_2$ be any point in $\dot\Sigma$. Then we claim that for all $q_1\in \dot\Sigma$ such that $q_1$ is not connected with $q_2$ by a diffractive bicharacteristic (in particular $q_1\neq q_2$), $(q_1,q_2)\notin \wf'_\b(S_+)$ and $(q_1,q_2)\notin \wf'_\b(S_-)$.

 In fact, if $t_1<t_2$   this simply follows from support properties of $S_+$. Otherwise, we can take $q_1'$ with $t'_1<t_2$  and $q_1'\dot\sim q_1$, then $(S_{q_1'}\times \{ q_2\})\cap \wf'_\b(S_+)=\emptyset$ by the previous argument, which then implies $(S_{q_1}\times \{ q_2\})\cap \wf'_\b(S_+)=\emptyset$ by the propagation of singularities statement in Theorem    \ref{opp} (strictly speaking, its obvious generalization microlocalized  outside of the diagonal, as $S_+$ is not a bi-solution everywhere), in particular $(q_1,q_2)\notin \wf'_\b(S_+)$. The analogous argument applies to $S_-$.
 
 Using \eqref{eq:wfbS} we conclude that 
 \beq\label{eq:wfS}
 \wf'_\b(S)\subset \{ (q_1,q_2) \st q_1 \dot\sim q_2 \}.
 \eeq
  Plugged into \eqref{eqlpm} this yields  \eqref{preciseWF2}. \qeds

  \begin{remark}Note that here we have proved an estimate
    \eqref{eq:wfS} on the singularities of the Dirac--Coulomb
    propagator $S$, but one could wonder if diffraction does actually
    occur and if one could not replace $q_1 \dot \sim q_2$ by
    $q_1\sim q_2$. The work \cite{Baskin2023} conjectures that this is
    \emph{not} the case (at least in the time-independent case), and
    diffraction does occur indeed, albeit the resulting extra
    singularities are $1-\epsilon$ order less severe (with
    $\epsilon>0$ arbitrary) in the sense of local Sobolev regularity.
    From our results we conclude that the diffractive singularities
    are harmless if one deals only with Hadamard states in the
    renormalization. On the other hand, diffraction implies a
    significantly more singular behaviour if one considers states that
    are not Hadamard (for instance, this issue arises when one
    subtracts the Dirac or Dirac--Coulomb vacuum in a $t$-dependent
    situation).  It could also potentially have bad implications for
    quantities renormalized with the Hadamard parametrix as in
    \cite{Marecki2003,Hollands2001b,Zahn2014,Schlemmer2015,Zahn2015a},
    though this still has to be studied in more detail.
\end{remark}

\setcounter{equation}{0}

\appendix\section{Second quantization for fermions}    \label{s:appA}
\subsection{Fermionic quasi-free states}  \label{s:appA1}
In this appendix we briefly recall basics of second quantization in the case of charged fermions following textbook references \cite{baez,thaller,derger,Bratteli1997} and explain the implications of our results for the  renormalized quantum current and charge density.  We use the convention that sesquilinear forms are linear in the second argument.

Given a complex inner product space $\cV$ with inner product $\bra \cdot, \cdot \ket$, the objective of second quantization is to construct a Hilbert space $\cH$ and  \emph{charged field operators}, i.e.~operator-valued anti-linear maps $\cV\ni  v \mapsto \psi(v)\in B(\cH)$ which satisfy the \emph{canonical anti-commutation relations} for $\cV$, i.e. 
\beq\label{eq:CAR} 
 \begin{aligned}
&\{ \psi^{*}(v),\psi^{*}(w) \}=\{ \psi(v),\psi(w) \}=0, \\
&\{ \psi(v), \psi^{*}(w) \}= \bra v,w \ket \one_\cH, \ \ v, w\in \cV.
\end{aligned} 
\eeq

It is advantageous to view the problem in a representation-theoretic
way, and first consider the unital $C^*$-algebra ${\rm CAR}(\cV)$
which is formally generated by abstract elements $\psi(v)$,
$\psi^*(w)$, $v,w\in \cV$, satisfying the canonical anti-commutation
relations \eqref{eq:CAR} (see e.g.~\cite[\S17.2.3]{derger} and
\cite[\S2.5]{derger} for the precise definition; note that it makes
reference to the neutral formalism).  Next, one chooses a {state}
$\omega$ on ${\rm CAR}(\cV)$, i.e., a positive continuous linear
functional of norm $1$. Then the \emph{GNS construction} (named after
Gelfand, Naimark and Segal) provides a CAR representation associated
to $\omega$. More precisely, the GNS construction produces a Hilbert
space $\cH_\omega$ using the inner product $(A,B)\mapsto \omega(A^* B)$ (the null
space is quotiented out to ensure non-degeneracy and then one takes the completion), and it constructs
field operators  (operator-valued anti-linear maps)  $\cV \ni v \mapsto \hat\psi(v)\in B(\cH_\omega)$ and a distinguished
{cyclic }\emph{vacuum vector} $\Omega\in \cH_\omega$ such that
$$
\omega(\psi(v)\psi^{*}(w))= \bra \Omega, \hat\psi(v)\hat\psi^{*}(w)  \Omega \ket_{\cH_\omega}, \ \ v,w\in \cV.
$$

In the sequel we restrict our attention to {quasi-free states}, defined as follows (we will always tacitly assume them to be gauge-invariant).  

\begin{definition}
A state $\omega$ on ${\rm CAR}(\cV)$ is \emph{gauge-invariant} if and only if 
\beq\label{eq:gi}
\omega(\psi^*(v_1)\cdots\psi^*(v_n)\psi(w_1)\cdots\psi(w_m))=0, \quad n\neq m
\eeq
for all $v_1,\dots,v_n,w_1,\dots,w_m\in \cV$. It is also \emph{quasi-free} if in addition 
\beq\label{eq:qf}
\omega(\psi^*(v_1)\cdots\psi^*(v_n)\psi(w_1)\cdots\psi(w_n))=\sum_{\sigma\in S_n}{\rm sgn}(\sigma)\prod_{j=1}^n \omega(\psi^*(v_j)\psi(w_{\sigma(j)})), 
\eeq
for all $n\in \nn$, $v_1,\dots,v_n,w_1,\dots,w_n\in \cV$, where $S_n$
is the set of permutations of $\{1,\dots, n\}$.
\end{definition}

Note that the RHS of \eqref{eq:qf} is simply the determinant of   $(\omega(\psi^*(v_i)\psi(w_{j})))_{ij}$.

Thus, taking also into account the canonical anti-commutation
relations \eqref{eq:CAR}, a quasi-free state $\omega$ is uniquely determined by the
sesquilinear form $\omega(\psi(v)\psi^{*}(w))$, $v,w\in \cV$. There
exists a rather general converse statement; here we only give a
special case where $\cV$ is assumed to be
complete. 

\begin{proposition}Suppose that $\cV$ is complete. If $\gamma\in
  B(\cV)$ is a density matrix, i.e.~satisfies $0\leq \gamma \leq
  \one$, then there exists a unique quasi-free state $\omega$ on ${\rm
    CAR}(\cV)$ such that
$$
{\omega(\psi(v)\psi^{*}(w))= \bra v, (\one-\gamma) w \ket, \quad
\omega(\psi^*(v)\psi(w))= \bra w, \gamma v \ket,
\quad v,w\in \cV.}
$$
\end{proposition}

Now in the context of the Dirac equation in external potentials
(possibly with Coulomb singularity as in \S\ref{ss:main}), we apply the above proposition in
the particular case when we are given a
density matrix at some fixed reference time $t_0$, $\gamma(t_0)\in B( L^2(\rr^3;\cc^4))$, and we denote by $\cH_{t_0}$ the Hilbert space and by 
$L^2(\rr^3;\cc^4) \ni v\mapsto\hat\psi_{t_0}(v)\in B(\cH_{{t_0}})$ the field operators
obtained by {the} GNS construction for
the associated quasi-free state $\omega_{t_0}$ on ${\rm CAR}(L^2(\rr^3;\cc^4))$.    These have the interpretation of  \emph{Cauchy data fields} at $t=t_0$. {We henceforth set $\cH:=\cH_{t_0}$.}

For the sake of illustration we give a more direct description of
$\hat\psi_{t_0}$ (strictly speaking, a unitarily equivalent CAR representation), under the extra assumption that $\gamma(t_0)$ is a
projection. We denote
$$
{\cZ_{t_0}^+=\Ran(\one-\gamma(t_0)),\quad
\cZ_{t_0}^-=\Ran\gamma(t_0),}
$$
{where the $\pm$ superscript labels the particle/anti-particle subspace.} In this case it is possible to obtain $\hat\psi_{t_0}$ as an operator-valued anti-linear map of the form
$$
\hat\psi_{t_0}(v)=a((\one-\gamma(t_0))v,0)+a^*(0,{\wbar{\gamma(t_0)v}})\in B(\cH_{t_0}),\qquad v\in L^2(\rr^3;\cc^4),
$$
acting on the \emph{fermionic Fock space}
$$
\cH_{t_0}=\textstyle\bigoplus_{n=0}^\infty {\textstyle\bigwedge}^n\big({\cZ_{t_0}^+\oplus\wbar{\cZ_{t_0}^-}}\big)
$$
over ${\cZ_{t_0}^+\oplus\wbar{\cZ_{t_0}^-}}$, where ${\wbar{\cZ_{t_0}^-}}$ is the image of ${\cZ_{t_0}^-}$ by the anti-linear involution denoted by $z\mapsto\wbar z$, which maps Cauchy data of solutions to Cauchy data of solutions of a charge reversed Dirac equation (it consists of taking the complex conjugate and then multiplying by $i\beta\alpha_2$). The wedge stands for anti-symmetrized tensor product. Above,
$$
{\cZ_{t_0}^+\oplus\wbar{\cZ_{t_0}^-}}\ni(z_+,\wbar{z_-})\mapsto a(z_+,\wbar{z_-}),a^*(z_+,\wbar{z_-})\in B(\cH_{t_0})
$$
are the fermionic annihilation and creation operators on $\cH_{t_0}$ (strictly speaking, families of operators), defined by $a^*(z_+,\wbar{z_-})\Psi=(z_+,\wbar{z_-})\wedge\Psi$. By convention ${\bigwedge}^0({\cZ_{t_0}^+\oplus\wbar{\cZ_{t_0}^-}})=\cc 1$ and then $\Omega:=1\in\cH_{t_0}$ is the vacuum vector. The general construction when $\gamma(t_0)$ is not necessarily a projection requires a further adjustment and is given by the \emph{Araki--Wyss representation}, see e.g.~\cite[\S17.2.3]{derger} and \cite[\S17.2.6]{derger}.

 {The fields $\hat\psi_{t_0}$ play the role of Cauchy data of
 spacetime fields. The Cauchy
 data fields at an
 arbitrary time $t$ are
 $$
 \hat\psi_t(v):=\hat\psi_{t_0}(U(t_0,t)v)\in B(\cH),
 \ \  v\in L^2(\rr^3;\cc^4),
 $$
 and correspond to the density
 $\gamma(t)=U(t,t_0)\gamma(t_0)U(t_0,t)$. For each $t$, $ \hat\psi_t(v)$  is an
 anti-linear operator-valued distribution in the spatial variables.
 The associated spacetime field is defined by
 \beq\label{eq:spf}
 C_\c^\infty(\rr^{1,3};\cc^4)\ni f\mapsto
 \hat\psi(f):=\hat\psi_{t_0}((Sf)(t_0))\in B(\cH).
 \eeq
Since $Sf$ is a solution of
 the Dirac equation, the operator-valued distribution $\hat\psi$ satisfies
 $D\hat\psi=0$.}

{Equivalently, one can use the algebraic formalism directly on
spacetime. Let $\cV$ be the completion of
$C_\c^\infty(\rr^{1,3};\cc^4)$, after quotienting by the null space, with
respect to the inner product
$$
f,g\mapsto \bra (Sf)(t_0),(Sg)(t_0)\ket_{L^2(\rr^3;\cc^4)}.
$$
Pulling back $\omega_{t_0}$ by $f\mapsto(Sf)(t_0)$ gives a quasi-free
state $\omega$ on ${\rm CAR}(\cV)$ whose covariances are
$$
\begin{aligned}
\omega(\psi(f)\psi^*(g))
 &=\bra (Sf)(t_0),(\one-\gamma(t_0))(Sg)(t_0)\ket_{L^2(\rr^3;\cc^4)}
 =\bra f,\Lambda^+g\ket_{L^2(\rr^{1,3};\cc^4)},\\
\omega(\psi^*(f)\psi(g))
 &=\bra (Sg)(t_0),\gamma(t_0)(Sf)(t_0)\ket_{L^2(\rr^3;\cc^4)}
 =\bra g,\Lambda^-f\ket_{L^2(\rr^{1,3};\cc^4)}.
\end{aligned}
$$
The fields obtained by GNS representation may be identified with those in \eqref{eq:spf}.}

Note that especially when discussing space-time quantities, most of the literature uses the relativistic Dirac--Coulomb  equation
 \beq\label{eq:Dir}
 \big( i^{-1} \gamma^\mu (\pa_\mu+i A_{\mu} ) +m\big) u=0,
 \eeq
 where $\gamma^0=\beta$, $\gamma^j= \beta \alpha^j$ are the gamma
 matrices and we used Einstein summation convention. Note that
 \eqref{eq:Dir} is related to the operator $D=i \p_t + H(t)$ used here by composition with the $\gamma^0$ matrix, so the
 natural space-time inner product for the operator in \eqref{eq:Dir}
 involves a $\gamma^0$ factor, but switching between the two operators is
 straightforward.

Now let $\omega_{\rm ref}$ be another quasi-free state on ${\rm CAR}(\cV)$ associated to some density matrix $\gamma_{\rm ref}(t)$. We assume that  $\gamma(t)$ and $\gamma_{\rm ref}(t)$ are Hadamard in the terminology of \S\ref{ss:main} (we make the same assumptions on the potentials $A_\mu$ as in \S\ref{ss:main}).

We denote by $x\in \rr^{1,3}$ points in spacetime.  Using the usual formal notation for distributions, in our conventions the \emph{relative  current} is the distribution-valued vector 
\beq\label{eq:current}
\bea {}
j^\mu_{\omega,\omega_{\rm ref}}(x)&=\big( \omega(\psi^*(x)\gamma^0\gamma^\mu \psi(y))- \omega_{\rm ref}(\psi^*(x)\gamma^0\gamma^\mu \psi(y))  \big)|_{y=x}, \\ 
&=\big( \bra \Omega, \hat\psi^*(x)\gamma^0\gamma^\mu \hat\psi(y) \Omega\ket_\cH- \omega_{\rm ref}(\psi^*(x)\gamma^0\gamma^\mu \psi(y))  \big)|_{y=x}, 
\eea
\eeq
where $f\mapsto \psi(f)$ are the abstract fields in the $C^*$-algebra
${\rm CAR}(\cV)$. The relative current plays the role of vacuum expectation value of the quantum version of the classical current
$$
j^\mu_u(x)= u^*(x)\gamma^0\gamma^\mu u(x),
$$
which becomes ill-defined if one naively inserts a quantum field
$\hat\psi$ (which is in general a highly singular distribution even if $A_\mu\equiv 0$) instead of a smooth solution $u$.
 On the other hand, our results imply that  
 $$
 j^\mu_{\omega,\omega_{\rm ref}}(x)\in C^\infty(\rr; L^1_{\rm loc}(\rr^3))\cap C^\infty(\rr\times(\rr^3\setminus \{0\})), \ \ \mu=0,\dots,3, $$
 is well-defined.
 
 We illustrate this on the example of the \emph{relative charge density}, i.e.~$\rho_{\gamma,\gamma_{\rm ref}}(x)=j^0_{\omega,\omega_{\rm ref}}(x)$ (emphasizing the dependence on the density matrices in the notation as in the introduction). Since $(\gamma^0)^2=I$ it can be written as
 $$
 \bea
 \rho_{\gamma,\gamma_{\rm ref}}(x)&= \big( \bra \Omega, \hat\psi^*(x) \hat\psi(y) \Omega\ket_\cH- \omega_{\rm ref}(\psi^*(x) \psi(y))  \big)|_{y=x} \\
 & = {\Tr_{\cc^4}\big(\Lambda^-(x,x)-\Lambda^-_{\rm ref}(x,x)\big)}, 
 \eea
$$
where  $\Lambda^-(x,y)$ is the spacetime Schwartz kernel (with values in 4$\times$4 matrices) of the operator
$$
\big(\Lambda^- f\big)(t)=\int_{-\infty}^\infty U(t,t_0)\gamma(t_0)  U(t_0,t_2) f(t_2)dt_2
$$
and $\Lambda^-_{\rm ref}$ is defined similarly using $\gamma_{\rm ref}(t_0)$. The on-diagonal restriction is well-defined by Propositions \ref{cor1} and \ref{prop:td}. Writing $x=(t,\bar{x})$, we get 
 $$
 \bea
 \rho_{\gamma,\gamma_{\rm ref}}(t,\bar{x})=\Tr_{\cc^4}(\gamma(t)- \gamma_{\rm ref}(t))(\bar{x},\bar{x}). 
 \eea
$$

  In a scattering situation one takes $\omega$ and $\omega_{\rm ref}$ to be the {\emph{out} and \emph{in} vacua, respectively}, and then  integrating $ \rho_{\gamma,\gamma_{\rm ref}}$ yields the charge created by switching on a time-dependent potential and then turning it off. 

\subsection{Wick ordering and quantum current} Let us now explain in what sense the relative current $j^\mu_{\omega,\omega_{\rm ref}}(x)$ is an expectation value
 of an operator-valued distribution called the \emph{renormalized} or \emph{Wick-ordered} (with respect to $\omega_{\rm ref}$) \emph{quantum current} $\leftw \psi^*\gamma^0\gamma^\mu\psi(x)\rightw_{\omega_{\rm ref}}$.

Let $\cH_{\rm fin}\subset\cH$ be the subspace spanned by vectors of the form
$$
{\hat\psi^*(g_1)\cdots\hat\psi^*(g_\ell)}\hat\psi(f_1)\cdots\hat\psi(f_k)\Omega,
\quad k{,\ell}\in\nn_0,\quad f_1,\dots,f_k{,g_1,\dots,g_\ell}\in C_\c^\infty(M;\cc^4).
$$
{By the CAR, every polynomial in fields and adjoint fields is a linear combination of such normal-ordered monomials; the cyclicity of $\Omega$ therefore implies that $\cH_{\rm fin}$ is dense.}

The pointwise products of distributions such as $\hat\psi\hat\psi^*$, $\hat\psi^*\hat\psi$, $\hat\psi^*\gamma^0\gamma^\mu\hat\psi$, etc., are ill-defined and need to be replaced by \emph{Wick squares} $\leftw\hat\psi\hat\psi^*\rightw_{\omega_{\rm ref}}$, $\leftw\hat\psi^*\hat\psi\rightw_{\omega_{\rm ref}}$, $\leftw\hat\psi^*\gamma^0\gamma^\mu\hat\psi\rightw_{\omega_{\rm ref}}$, etc. We focus on explaining the definition of $\leftw\hat\psi^*\hat\psi\rightw_{\omega_{\rm ref}}$. First, one introduces the \emph{Wick monomial}
$$
\leftw\hat\psi^*(f)\hat\psi(g)\rightw_{\omega_{\rm ref}}=\hat\psi^*(f)\hat\psi(g)-\omega_{\rm ref}\big(\psi^*(f)\psi(g)\big)\one_\cH,
$$
or in function-like notation for distributions,
\beq
\leftw\hat\psi^*(x)\hat\psi(y)\rightw_{\omega_{\rm ref}}=\hat\psi^*(x)\hat\psi(y)-\omega_{\rm ref}\big(\psi^*(x)\psi(y)\big)\one_\cH.
\eeq
Then $\leftw\hat\psi^*\hat\psi(x)\rightw_{\omega_{\rm ref}}$ is defined by restricting $\leftw\hat\psi^*(x)\hat\psi(y)\rightw_{\omega_{\rm ref}}$ to the diagonal $x=y$. If well-defined, this in particular gives the vacuum expectation value
$$
\bea
\bra\Omega,\leftw\hat\psi^*\hat\psi(x)\rightw_{\omega_{\rm ref}}\Omega\ket_\cH
&=\big(\bra\Omega,\hat\psi^*(x)\hat\psi(y)\Omega\ket_\cH-\omega_{\rm ref}(\psi^*(x)\psi(y))\big)|_{y=x}\\
&=\rho_{\gamma,\gamma_{\rm ref}}(x).
\eea
$$

\begin{proposition} Under the hypotheses in \emph{\S\ref{ss:main}}, the Wick square $\leftw\hat\psi^*\hat\psi(x)\rightw_{\omega_{\rm ref}}$ is well defined as a locally integrable function on $\rr^{1,3}$ with values in quadratic forms on $\cH_{\rm fin}$.
\end{proposition}
\proof Let $\Psi\in\cH_{\rm fin}$ be a normal-ordered monomial as above. We compute using the quasi-free property
$$
\bea {}
\bra\Psi,\hat\psi^*(f)\hat\psi(g)\Psi\ket_\cH
&=\bra\Omega,\hat\psi^*(f)\hat\psi(g)\Omega\ket_\cH\|\Psi\|_\cH^2
+{\sum_{i,j}c_{i,j}F_i(f)G_j(g)},
\eea
$$
where $c_{i,j}$ are linear combinations of products of two-point functions not involving $f$ nor $g$, and {up to complex conjugation, $F_i$ and $G_j$ are singly smeared two-point functions involving $\Lambda^+$ or $\Lambda^-$ that have one of the following forms:
$$
\begin{aligned}
F_i(f)&=\bra\Omega,\hat\psi^*(f)\hat\psi(f_i)\Omega\ket_\cH
      =\bra f_i,\Lambda^-f\ket_{L^2},
&  \ 
G_j(g)&=\bra\Omega,\hat\psi^*(g_j)\hat\psi(g)\Omega\ket_\cH
      =\bra g,\Lambda^-g_j\ket_{L^2},\\
F_i(f)&=\bra\Omega,\hat\psi(f_i)\hat\psi^*(f)\Omega\ket_\cH
      =\bra f_i,\Lambda^+f\ket_{L^2},
&
G_j(g)&=\bra\Omega,\hat\psi(g)\hat\psi^*(g_j)\Omega\ket_\cH
      =\bra g,\Lambda^+g_j\ket_{L^2}.
\end{aligned}
$$
Thus, in function-like notation, $F_i(x)$ and $G_j(x)$ are, up to complex conjugation, of the form $(\Lambda^\pm h)(x)$ for some test function $h$.} Thus, mixing function-like and distributional notation,
\beq\label{eq:exp}
\bea
\bra\Psi,\leftw\hat\psi^*\hat\psi(x)\rightw_{\omega_{\rm ref}}\Psi\ket_\cH
&=\|\Psi\|_\cH^2\Tr_{\cc^4}\big(\Lambda^-(x,y)-\Lambda^-_{\rm ref}(x,y)\big)|_{x=y}\\
&\phantom{=\,}+{\sum_{i,j}c_{i,j}F_i(x)G_j(x)}.
\eea
\eeq
The first term is well-defined by Propositions \ref{cor1} and \ref{prop:td}. For the sum over $i,j$, regularity gives
$$
{F_i,G_j\in H^{1,\infty}_{\b,\loc}(M)\subset H^1_{\loc}(M)\subset L^4_{\loc}(M),}
$$
so ${F_iG_j\in L^2_{\loc}(M)\subset L^1_{\loc}(M)}$. Thus the products in \eqref{eq:exp} are well-defined as locally integrable functions. {Finite linear combinations of normal-ordered monomials are handled by expanding the quadratic form.} \qed

The discussion of the renormalized quantum current $\leftw\hat\psi^{*}\gamma^0\gamma^\mu\hat\psi(x)\rightw_{\omega_{\rm ref}}$ is entirely  analogous.

\medskip

{\small
\subsubsection*{Acknowledgments} The authors would like to thank
Christian Gérard and \'Eric Séré  for useful discussions.  DB
acknowledges partial support from NSF grant DMS--1654056.  MW acknowledges support from the
ANR-20-CE40-0018 grant. JW acknowledges partial support from NSF grant
DMS--2054424 and from Simons Foundation grant MPS-TSM-00007464;
he is also grateful for the hospitality of 
   Utrecht University in May 2025. The authors acknowledge support from the Institut Henri Poincaré (UAR 839 CNRS-Sorbonne Université), LabEx CARMIN (ANR-10-LABX-59-01), and  Erwin Schrödinger International Institute for Mathematics and Physics (ESI). \medskip }
  \bibliographystyle{abbrv}
  \bibliography{diraccoulomb}
 
\end{document}